%% file: main.tex
\documentclass[11pt]{article}

\input{header}

\title{Cactus Representations in Polylogarithmic Max-flow\\via Maximal Isolating Mincuts}
\author{
Zhongtian He\\
Princeton University
\and
Shang-En Huang\thanks{Supported by NSF Grant No. CCF-2008422.}\\
Boston College
\and
Thatchaphol Saranurak\thanks{Supported by NSF CAREER grant 2238138.}\\
University of Michigan
}
\date{}

\begin{document}

\maketitle

\pagenumbering{gobble}
\input{0-abstract}
\newpage
\pagebreak

\tableofcontents
\pagebreak

\pagenumbering{arabic}
\input{1-intro}

\section{Preliminaries}\label{sec:preliminaries}

Let $G=(V, E, w)$ be an undirected, connected, and weighted graph with positive edge weights $w:E\to\mathbb{R}^+$. 
A cut is a partition $(X, V\setminus X)$, when the context is clear we use $X$ (or $V\setminus X$) representing the cut $(X, V\setminus X)$ for brevity.
For any subsets $X, Y\subseteq V$ we define the \emph{value} between $X$ and $Y$ to be  $\C(X, Y) = \sum_{u\in X, v\in Y} w(u, v)$.
When $Y=V\setminus X$ we simply denote $\C(X, V\setminus X)$ by $\C(X)$.
A nonempty subset $X\subset V$ is said to be a (global) \emph{mincut} if $\C(X)=\min_{S: \emptyset \subset S \subset V}\C(S)$.
The set of boundary edges incident to $X$ is denoted by $\partial X$.
For two disjoint subsets of vertices $A$ and $B$, a cut $(X, V\setminus X)$ \emph{separates} $A$ and $B$ if $X\subseteq A$ and $Y\subseteq B$.

We say two subsets of vertices $X$ and $Y\subseteq V$ are \emph{crossing} if  all of $X\cap  Y$, $X\setminus Y$, $Y\setminus X$, and $V\setminus(X\cup Y)$ are nonempty.
The value function $\C$ satisfies \emph{submodularity} and \emph{posi-modularity} (see \cite{NI00}) on crossing subsets.

\begin{lemma}[{\cite{NI00}}]\label{lem:submod}
Let $X, Y\subseteq V$ be subsets that are crossing. Then,
\begin{itemize}[itemsep=0pt]
    \item (Submodularity) $\C(X) + \C(Y)\ge \C(X\cap Y) + \C(X\cup Y)$, and
    \item (Posi-modularity) $\C(X) + \C(Y)\ge \C(X\setminus Y) + \C(Y\setminus X)$. \hfill $\square$
\end{itemize}
\end{lemma}

Let $X\subseteq V$ be any subset of vertices, we define the \emph{contracted graph} $G/X$ to be the graph obtained from $G$ by (1) contracting all vertices in $X$ into one vertex, (2) replacing all multi-edges with a single edge with the corresponding edge weight, and finally (3) removing all self-loops.

\paragraph{Steiner Mincuts and $t$-Isolating Cuts.}
Let $\T\subseteq V$ be a terminal vertex set of size at least 2.
For any proper subset $A\subsetneq \T$, an \emph{$A$-cut} of $\T$ is a cut that separates $A$ and $\T\setminus A$. An \emph{$A$-mincut} of $\T$ is a minimum valued $A$-cut of $\T$, denoted by the vertex set $X_A$.
Conveniently, for any vertex $t\in \T$ we define a cut $X_t$ to be \emph{$t$-isolating mincut} of $\T$ if $X_t$ is a $\{t\}$-mincut of $\T$.
A \emph{$\T$-Steiner mincut} is defined to be a minimum valued $A$-mincut among all proper subsets $A\subsetneq \T$. The value of a $\T$-Steiner mincut on the graph $G$ is denoted as $\lambda_G(\T)$.

$A$-mincuts of $\T$ are not necessarily unique.
Fortunately, with the submodularity and posi-modularity mentioned in~\Cref{lem:submod},
these $A$-mincuts behave similarly to global mincuts
in the sense that
there is a unique \emph{minimal} and \emph{maximal} $A$-mincut of $\T$.

\begin{definition}[Maximal and Minimal $A$-Mincuts of $\T$]
\label{def:maximal and minimal A-mincut}
We say that an $A$-mincut $X_A$ of $\T$ is \emph{maximal} (resp. minimal), if for any other $A$-mincut $X'_A$ of $\T$, we have $X_A\supseteq X'_A$ (resp. $X_A\subseteq X'_A$).
\end{definition}

\paragraph{Isolating Mincuts.}
Given a graph $G$ and a set of terminal vertex $\T$, the \emph{isolating mincuts of $\T$} is a collection of $t$-isolating mincuts that separates each single terminal vertex $t\in \T$ from $\T\setminus\{t\}$. 
Li and Panigrahi~\cite{LP20} gave an algorithm that computes a minimal isolating cut efficiently.

\begin{theorem}[{\cite[Isolating Cut Lemma]{LP20}}]\label{thm:isocut}
Given a graph $G=(V, E)$ and any terminal set $\T \subseteq V$, there is an algorithm that returns the minimal $t$-isolating mincuts for all $t \in \T$ in $O(\log|\T|\cdot\MaxFlow(2n, 2m))$ time.
\end{theorem}

The function $\mathrm{MaxFlow}(x, y)$ denotes the time needed for solving an $st$-maxflow instance with $x$ vertices and $y$ edges.
We assume that the function $\mathrm{MaxFlow}(x, y)$ is $\Omega(x+y)$, is non-decreasing, and is superadditive.\footnote{An integral function $f(x, y)$ is said to be \emph{superadditive} if for any $x_1, y_1, x_2, y_2\in\mathbb{N}$ we have $f(x_1+x_2, y_1+y_2)\ge f(x_1, y_1)+f(x_2, y_2)$. The function is \emph{non-decreasing} if both $f(x+1, y)\ge f(x, y)$ and $f(x, y+1)\ge f(x, y)$ holds.}

\section{Maximal Isolating Mincuts}
\label{sec:max isocut}

The maximal minimum isolating cut problem states that, given a graph $G=(V,E)$ with $n$ vertices and $m$ edges and a terminal set $\T$, the goal is to obtain the maximal isolating mincut $X_v$ for each terminal $v\in \T$.
In this section, we give an efficient algorithm for solving the maximal isolating mincuts problem in almost linear time, which summarizes as the following \Cref{thm:max-min-iso-cut}.

\begin{theorem}\label{thm:max-min-iso-cut}
There exists an algorithm that, given an undirected weighted graph $G=(V, E)$ and a terminal set $\T\subseteq V$, in $O(\log|\T|\cdot\mathrm{MaxFlow}(3n,4m))$ time computes the maximal isolating mincut of all terminals $v\in \T$ with respect to $\T$.
\end{theorem}

\paragraph{Key Insight.}
The crux of our maximal isolating mincut algorithm is an observation\footnote{A similar observation can also be found in Dinitz and Vainshtein~\cite[3-Star Lemma]{DV94}.} to any set of three pairwise crossing  mincuts to three disjoint subsets of $\T$ --- their intersection is always an empty set due to \emph{only} posi-modularity but not submodularity.

To formally prove this, we first state a standard fact about posi-modularity in~\Cref{lem:disjoint-posi-modularity} (see its proof in \Cref{app:proof of dosjoint posi}). Then, we state the key observation as the Pairwise Intersection Only Lemma below.

\begin{lemma}[Disjoint \& Posi-modularity]\label{lem:disjoint-posi-modularity}
Let $A, B\subseteq \T$ be two nonempty subsets of terminals with $A\cap B=\emptyset$. Let $X_A$ (resp. $X_B$) be an $A$-mincut (resp. $B$-mincut) of $\T$. Then, $X_A\setminus X_B$ is an $A$-mincut of $\T$, and $X_B\setminus X_A$ is a $B$-mincut of $\T$. 
\end{lemma}

\begin{lemma}[Pairwise Intersection Only]
\label{lem:three-crossing-isolating-mincuts}
Given a graph $G$, let $A, B, C\subseteq\T$ be three disjoint nonempty subsets of terminals. 
Let $X_A, X_B, X_C\subseteq V$ be any $A$-mincut, $B$-mincut, and $C$-mincut of $\T$ respectively.
Then $X_A\cap X_B\cap X_C=\emptyset$.
\end{lemma}

\begin{proof}
The proof is only interesting whenever the intersection of any two mincuts is non-empty (i.e. crossing). Thus, without loss of generality, we assume that $X_A\cap X_B\neq\emptyset$, $X_B\cap X_C\neq\emptyset$, and $X_C\cap X_A \neq \emptyset$.

Define $X'_A = (X_A\setminus X_B)\setminus  X_C$, $X'_B = (X_B\setminus X_C)\setminus X_A$ and $X'_C = (X_C\setminus X_A)\setminus X_B$. These sets are non-empty since $A\subseteq X'_A$, $B\subseteq X'_B$, and $C\subseteq X'_C$. By posi-modularity (\Cref{lem:disjoint-posi-modularity}) we know that $X'_A$, $X'_B$, and $X'_C$ are also $A$-mincut, $B$-mincut, and $C$-mincut of $\T$ respectively, and thus
\begin{equation}\label{eq:crossing-formula}
    \underbrace{(\C(X_A)+\C(X_B)+\C(X_C)) - (\C(X'_A)+\C(X'_B)+\C(X'_C))}_{(\mathrm{LHS})} = 0. \tag{$\star$}
\end{equation}

Assume for contradiction that $X := X_A\cap X_B\cap X_C \neq \emptyset$.
We now claim that there is no edge from $X$ to $X_A\setminus X$. 
Note that with the assumption that $G$ is connected, we have $\C(X, V\setminus X_A)>0$.
The claim implies that 
\[\C(X_A) = \C(X_A \setminus X) - \C(X_A \setminus X, X) + \C(X,V\setminus X_A) > \C(X_A\setminus X),\]
which contradicts the fact that  $X_A$ is $A$-mincut.

We shall prove the claim by two steps: (1) every edge in the graph contributes to non-negative values to the left-hand side (LHS) of \eqref{eq:crossing-formula}. Therefore,
any edge that contributes to a positive amount to LHS should not exist in $G$ by \Cref{eq:crossing-formula}.
(2) For any boundary edge $(u, v)\in \partial X$ with $u\in X$ and $v\notin X$, we have $v\notin X_A\cup X_B\cup X_C$, which implies the claim.

\begin{enumerate}
    \item We consider the cases that $e$ contributes to how many terms of $\C(X'_A), \C(X'_B)$ and $\C(X'_C)$. Since $X'_A$, $X'_B$, and $X'_C$ are disjoint and an edge $e$ incidents to two vertices, $e$ contributes to at most two of the three terms. There is nothing to prove if $e$ contributes to 0 of them.

    If $e$ contributes to one of them, without loss of generality $\C(X'_A)$. Then $e$ contains some vertex in $X'_A$, and also some vertex $v$ in $V\setminus X'_A = (V\setminus X_A)\cup X_B\cup X_C$. Therefore $v$ is in either $V\setminus X_A$, $X_B$ or $X_C$, hence also contributes to either $\C(X_A), \C(X_B)$ or $\C(X_C)$ respectively.
    
    Otherwise $e$ contributes to two of them, without loss of generality $\C(X'_A)$ and $\C(X'_B)$, then $e$ contains some vertex $u\in X'_A$ and $v\in X'_B$. By the definition of $X'_A$ and $X'_B$, we have $u\in X_A$, $u\notin X_B$, $v\in X_B$, $v\notin X_A$.  So $e$ also contributes to $\C(X_A)$ and $\C(X_B)$.
    \item We prove (2) using (1) by showing that the contribution to LHS is positive if $u\in X$ and $v\in (X_A\cup X_B\cup X_C)\setminus X$. Again we consider the cases that $e$ contributes to how many terms of $\C(X'_A), \C(X'_B)$, and $\C(X'_C)$. Note that $e$ can not contribute to more than one of them since $X'_A$, $X'_B$, and $X'_C$ are disjoint and one of the endpoints $u\in X$ is in neither of them.

    If $e$ contributes to none of them, then we need to show that $e$ contributes to at least one of $\C(X_A), \C(X_B)$, and $\C(X_C)$. This claim is directly from the fact that $X = X_A\cap X_B\cap X_C$, $u\in X$, and $v\notin X$.

    Otherwise, $e$ contributes to one of them, without loss of generality $\C(X'_A)$. By the definition of $X'_A$, we have $v\notin X_B$ and $v\notin X_C$. Therefore $e$ also contributes to $\C(X_B)$, and $\C(X_C)$. \qedhere
\end{enumerate}
\end{proof}

\paragraph{Bounding the output size.} Before we describe our algorithm, we emphasize that the total size of all maximal $v$-isolating mincuts is $O(n)$, as a consequence of \Cref{lem:three-crossing-isolating-mincuts}:

\begin{lemma}\label{lem:total-size-max-min-iso-cut}
Let $G$ be a graph and $\T$ be a set of terminals.
For each $v\in \T$, let $X_v$ be any $v$-isolating mincut.
Then, $\sum_{v\in \T} |X_v| \le 2n$.
\end{lemma}

\begin{proof}
Every vertex $u\in V$ belongs to at most two isolating mincuts.
Suppose by contradiction that there exist three terminal vertices $v_1, v_2, $ and $v_3$ such that $u\in X_{v_1}\cap X_{v_2}\cap X_{v_3}$. By \Cref{lem:three-crossing-isolating-mincuts} we know that $X_{v_1}\cap X_{v_2}\cap X_{v_3}=\emptyset$, a contradiction. Hence, a counting argument shows that $\sum_{v\in T}|X_v|\le 2n$.
\end{proof}

We remark that for our purpose we apply \Cref{lem:total-size-max-min-iso-cut} where $X_v$ is the maximal $v$-isloating mincut for all $v\in\T$.
\Cref{lem:total-size-max-min-iso-cut} implies that the total output size is $O(n)$. Thus, all maximal $v$-isolating mincuts can be stored explicitly.
Now, we are ready for introducing the algorithm.

\paragraph{A Divide and Conquer Algorithm.}
The existing algorithms~\cite{LP20, CQ21} for finding minimal isolating mincuts 
use \emph{one-shot recursion} by first computing $O(\log|\T|)$ cuts on $G$.
These cuts partition the vertex set into $O(|\T|)$ subsets, and isolate every terminal from $\T$.
The minimal $t$-isolating mincuts can then be obtained within each part (and contracting everything outside the part).
This does not work for us!
The most evident reason is that the maximal $t$-isolating mincuts may not be disjoint.

Instead of using one-shot recursion, we can also consider a divide and conquer algorithm that is equivalent to the existing algorithms.
The algorithm considers one cut at a time.
Each cut splits the terminal set (and the vertex set) into two parts, creating two subproblems.
Each subproblem is obtained by contracting all vertices on one side of the cut into a single vertex.
The subsequent cuts are now affecting both subproblems, splitting each subproblem into another two subproblems, and so on and so forth.

Our algorithm solves the maximal isolating mincut problem using a divide and conquer approach very similar to the algorithm described above, but with a slight twist.
In our algorithm,
every subproblem is derived from $G$ via contractions,
and 
there will be \emph{at most one} contracted vertex in each subproblem.
We call such a special vertex $p$ the \emph{pivot} of a subproblem.
Notice that $p=\texttt{null}$ when the algorithm first enters the recursion, as there is no contracted vertex.
Throughout the execution,
the algorithm recursively splits the terminal set $\T\setminus \{p\}$ arbitrarily into two halves $\T\setminus\{p\}=A\cup B$.
Then, the algorithm computes $X_A$ (the maximal $A$-mincut of $\T$) and $X_B$ (the maximal $B$-mincut of $\T$)
by solving \emph{two} $s$-Maximal $st$-Mincut instances (which can be solved using one $st$-MaxFlow and a linear time post processing).

We observe (via submodularity, see \Cref{lem:steiner-modularity}) that for each terminal $v\in A$, the maximal $v$-isolating mincut $X_v$ must be fully contained in $X_A$, i.e., $X_v\subseteq X_A$.
Hence, in order to find out the maximal $v$-isolating mincuts for all vertices $v\in A$, it is safe to contract everything outside $X_A$ and recursively compute maximal isolating mincuts on the contracted graph $G_A := G/(V\setminus X_A)$.
Notice that the pivot $p$ as well as all terminal vertices in $B$ are outside of $X_A$, 
so within the contracted graph $G_A$ the contracted vertex $p_A$ shall be considered as a terminal vertex in the subproblem.

Similarly, the algorithm contracts everything outside $X_B$ which leads to the pivot vertex $p_B$, and recursively computes maximal isolating mincuts on $G/(V\setminus X_B)$ too.
This divide and conquer procedure continues until the number of terminal vertices becomes a constant, where the maximal $v$-isolating mincut can be computed for each vertex $v\in \T$ individually using a max-flow.
The algorithm is summarized in \Cref{alg:max-min-iso-cut}.

\begin{algorithm}
\caption{Maximal isolating mincuts.}
\label{alg:max-min-iso-cut}
\begin{algorithmic}[1]\small
\Procedure{MaxIsoMincut}{$G, \T$}
\State Call \Call{MaxIsoMincutWithPivot}{$G, \T, {\texttt{null}}$}.
\EndProcedure
\Procedure{MaxIsoMincutWithPivot}{$G, \T, p$}
\If{$|\T|\le 4$}\Comment{Base case.}
\State Obtain the maximal $v$-isolating mincut $X_v$ for each $v\in \T$.\Comment{Use MaxFlow.}\label{line:run-max-flow}
\Else 
\State Partition $\T\setminus\{p\}$ arbitrarily into two similarly sized sets $A\cup B=\T\setminus\{p\}$.
\State Compute $X_A$, the maximal $A$-mincut of $\T$. \Comment{Use MaxFlow.}
\State Compute $X_B$, the maximal $B$-mincut of $\T$. \Comment{Use MaxFlow.}
\State Let $G_A \gets G/(V\setminus X_A)$ where $V\setminus X_A$ gets contracted to $p_A$.\label{line:contract-ga}
\State Let $G_B \gets G/(V\setminus X_B)$ where $V\setminus X_B$ gets contracted to $p_B$.
\label{line:contract-gb}
\State Invoke \Call{MaxIsoMincutWithPivot}{$G_A, A\cup \{p_A\}, p_A$}.
\State Invoke \Call{MaxIsoMincutWithPivot}{$G_B, B\cup \{p_B\}, p_B$}.
\EndIf
\EndProcedure
\end{algorithmic}
\end{algorithm}

\paragraph{Analysis}
We now prove the correctness and the runtime.
As mentioned above, the correctness of \Cref{alg:max-min-iso-cut} depends on a standard fact on submodularity on two nesting subsets (for completeness see its proof in
\Cref{app:proof of nesting submod}):

\begin{lemma}[Nesting \& Submodularity]
\label{lem:steiner-modularity}
Let $G$ be a graph and let $\T$ be the set of terminals.
 Consider two nonempty subsets $A$ and $B$ of terminals such that $A\subseteq B\subsetneq \T$.
Let $X_A$ (resp. $X_B$) be any $A$-mincut (resp. $B$-mincut) of $\T$.
Then, $X_A\cap X_B$ is a $A$-mincut of $\T$.
Respectively, $X_A\cup X_B$ is a $B$-mincut of $\T$.
\end{lemma}

Consider the recursion tree throughout executing \Cref{alg:max-min-iso-cut}.
We say that a subproblem is a \emph{leaf} if \Cref{line:run-max-flow} is executed for the subproblem and thus no further recursions are invoked.

\begin{lemma}[Correctness]\label{lem:max-iso-mincut-correctness}
Fix a terminal vertex $v\in \T$. There is a unique (leaf) subproblem where the maximal isolating mincut for $v$ is computed.
Let $\hat{X}_v$ be the cut returned at Line~\ref{line:run-max-flow} for vertex $v$.
Then $\hat{X}_v=X_v$ is the maximal $v$-isolating mincut on the graph $G$.
\end{lemma}

\begin{proof}
Fix $v\in \T$.
Since in each divide step, exactly one subproblem contains $v$ so there will be exactly one leaf subproblem where Line~\ref{line:run-max-flow} is invoked for $v$.
Clearly $\hat{X}_v$ does not contain any contracted vertex (since the only contracted vertex in the subproblem is the pivot.)
Now, to show that $\hat{X}_v=X_v$, 
it suffices to show that in each divide step,
all vertices in $X_v$ are not contracted. Indeed, by submodularity in \Cref{lem:steiner-modularity},
suppose without loss of generality that $v\in A$ then we must have $X_v\subseteq X_A$ (otherwise $X_A$ is not a maximal isolating mincut of $A$ since $X_v\cup X_A$ is a larger-sized isolating mincut of $A$.)
\end{proof}

\begin{lemma}[Runtime]\label{lem:max-iso-mincut-runtime}
\textsc{MaxIsoMincut}$(G, \T)$ runs in $O(\log |\T|\cdot \mathrm{MaxFlow}(3n, 4m))$ time.
\end{lemma}

\begin{proof}
It suffices to bound the sum of graph sizes in all subproblems throughout the execution of \Cref{alg:max-min-iso-cut}.
First of all, the maximum depth of the recursion tree is 
$\left\lceil \log|\T| \right\rceil $ since in each recursive call the number of non-pivot terminals is reduced to half.
In addition, the number of subproblems in each recursion depth $i$ is at most $\min\{2^i, |\T|\}$.

Now, we focus on a particular recursion depth $i>0$. Let $\{(G_j, \T_j, p_j)\}$ be all the subproblems whose recursion depth is $i$.
We observe that all terminals except pivot go to exactly one subproblem so the subsets $\T'_j := \T_j\setminus \{p_j\}$ are disjoint.
Moreover, by Lines~\ref{line:contract-ga}-\ref{line:contract-gb} we know that for each $j$, removing the pivot $p_j$ from $G_j$ it is exactly the maximal $\T'_j$-mincut of $\T$ on $G$.

Using the Pairwise Intersection Only Lemma (\Cref{lem:three-crossing-isolating-mincuts}),
we are able to conclude that every vertex in the input graph $G$ occurs in at most two subproblems at recursion depth $i$.
Therefore, the total number of vertices across all subproblems at depth $i$ is $\sum_{j}|V(G_j)| \le 2n + \min\{2^i,|\T|\}\le 2n+|\T| \le 3n$.
To analyze
the total number of edges across all subproblems at depth $i$,
we notice that each edge has at least one non-pivot endpoint.
Thus, the total number of edges can be bounded by the sum of all vertex degrees
$\sum_j|E(G_j)| \le 4m$.

Finally, in each subproblem $(G', \T')$, where the graph $G'$ has $n'$ vertices and $m'$ edges, the algorithm computes the maximal $A$-mincut of $\T$ (denoted by $X_A$) and the maximal $B$-mincut of $\T$ (denoted by $X_B$) by the following steps.
First, the algorithm creates a flow graph $G'_{\mathrm{flow}}$ where the source vertex $s$ is obtained by merging all vertices in $A$ and the sink vertex $t$ is obtained by merging all vertices in $B$.
This graph has at most $n'$ vertices and at most $m'$ (undirected) edges since $G'$ is formed from the input graph $G$ by a sequence of contraction.
Then, the algorithm finds any $st$-MaxFlow $f$ within $\mathrm{MaxFlow}(n', m')$ time.
Finally, the algorithm examines the residual graph $G'^{(f)}_{\mathrm{flow}}$ in $O(n'+m')$ time: $X_A$ is exactly the set of vertices that do \emph{not} reach $t$ and $X_B$ is exactly  the set of vertices that are \emph{not} reachable from $s$.

Therefore, by summing up the runtime per recursion depth, we obtain an upper bound to the desired total runtime $O(\log|\T|\cdot\mathrm{MaxFlow}(3n, 4m))$.
\end{proof}

\begin{proof}[Proof of \Cref{thm:max-min-iso-cut}.]
\Cref{thm:max-min-iso-cut} follows directly by \Cref{alg:max-min-iso-cut}, \Cref{lem:max-iso-mincut-correctness}, and \Cref{lem:max-iso-mincut-runtime}.
\end{proof}

\section{Steiner Cactus Construction}
\label{sec:steiner-cactus}

In this section, we apply the maximal isolating mincut algorithm from \Cref{sec:max isocut} to construct a a \emph{Steiner cactus representation} that succinctly represents all $\T$-mincuts on a given graph $G$ with terminal set $\T$ using $O(|\T|)$ space.

\begin{definition}[Steiner Cactus, see also~\cite{DV94,CX17}]
\label{def:steiner-cactus}
Given a graph $G$ and a terminal set $\T$, a \emph{$\T$-Steiner cactus} $(H, \phi)$ is a weighted cactus graph $H$ with a mapping $\phi:\T\to V(H)$ such that 
(1)
edges in a cycle have weights $\lambda_G(\T)/2$ and edges not in a cycle have weights $\lambda_G(\T)$,
(2)
an $A$-mincut of $\T$ is a $\T$-Steiner mincut if and only if a global mincut on $H$ separates $\phi(A)$ and $\phi(\T\setminus A)$.
\end{definition}
We say that a node $v\in V(H)$ in a cactus is \emph{non-empty} if there exists a terminal $t$ where $\phi(t)=v$. There may exist \emph{empty} nodes in $H$.
Notice that cactus and Steiner cactus representations of a graph may not be unique.\footnote{See the work of Nagamochi and Kameda~\cite{nagamochi1994canonical}
for canonical cactus representations of a graph.}
Our Steiner cactus algorithm is now summarized below as \Cref{thm:steiner-cactus-main}.

\begin{theorem}\label{thm:steiner-cactus-main}
Let $G$ be a graph with $n$ vertices and $m$ edges.
Let $\T$ be a set of terminals.
There exists a randomized Monte Carlo algorithm such that,
with probability $1-8n^{-10}$, the algorithm correctly computes a $\T$-Steiner cactus
in $O((\log^4 n)\cdot \mathrm{MaxFlow}(3n, 4m+8n\log |\T|))$ time.
\end{theorem}

\subsection{Divide and Conquer Approach: Prior Works}
\label{sec:divide-and-conquer prior works}

Chekuri and Xu's algorithm~\cite{CX17} finds a linear-sized \emph{hypercactus representation} that represents all (global) mincuts on a hypergraph.
This hypercactus representation degenerates to a cactus representation on a normal graph and also in the Steiner setting.
We now briefly describe their framework in terms of constructing a $\T$-Steiner cactus.

The main idea of Chekuri and Xu's algorithm is to successively find \emph{$\T$-splits} --- $\T$-mincuts that have at least two terminal vertices on both sides.
After obtaining a $\T$-split $(X, V\setminus X)$, a \emph{simple refinement} conceptually \emph{decomposes} the graph $G$ into two graphs $G_1$ and $G_2$, where each of them is obtained from a copy of $G$ with all vertices from one side (either $X$ or $V\setminus X$) are contracted.
Notice that the contracted vertices will be treated as terminals in the decomposed graphs, which we call \emph{anchor vertices}.

\begin{definition}\label{def:t-split}
A \emph{split} (or a \emph{$\T$-split}) of a terminal set $\T$ on a graph $G$ is a $\T$-mincut $(X, V\setminus X)$ such that both $X$ and $V\setminus X$ have at least two terminal vertices, i.e., $|\T \cap X|, |\T \setminus X| \ge 2$.
\end{definition}

\begin{definition}[Simple Refinement and Anchor Vertex]
Fix a graph $G=(V, E)$ and a terminal set $\T$.
We say that
    $\{(G_1, \T_1), (G_2, \T_2)\}$ is a \emph{simple refinement} of $G$ if $G_1$ and $G_2$ are graphs obtained through a $\T$-split $(X, V\setminus X)$ of $G$ and a new \emph{anchor vertex}\footnote{In \cite{CX17} the authors called these vertices \emph{marker vertices}.} $a$ as follows.
    \begin{itemize}[parsep=0pt]
        \item $G_1 := G/ (V\setminus X)$ such that $V\setminus X$ gets contracted to $a$.
        \item $\T_1 := (\T\cap X)\cup \{a\}$.
        \item $G_2 := G/ X$ such that $X$ gets contracted to $a$.
        \item $\T_2 := (\T\cap (V\setminus X))\cup \{a\}$.
    \end{itemize}
\end{definition}

Chekuri and Xu's algorithm maintains a \emph{decomposition} $\G=\{(G_i, \T_i)\}$ (initialized with the input graph $\{(G, \T)\}$),
iteratively finds a $\T_i$-split to any graph $(G_i, \T_i)\in\G$, and replaces the graph $G_i$ with its simple refinements.
Since each simple refinement creates an anchor vertex that appears in both decomposed graphs, at any time, the decomposition $\G$ admits a \emph{decomposition tree}.

\begin{definition}[Decomposition]
Fix a graph $G$ and a terminal set $\T$.
A \emph{decomposition} $\mathcal{G}=\{(G_i, \T_i)\}$ is a collection of graphs and terminal vertices obtained by performing an arbitrary sequence of simple refinements.
\end{definition}

Finally, the algorithm halts when no splits exist in the current decomposition. In this case, the decomposition is called a \emph{prime decomposition}.

\begin{definition}[Prime Decomposition]
Fix a graph $G$ and a terminal set $\T$.
We say that a decomposition $\mathcal{G}=\{(G_i, \T_i)\}$ is \emph{prime} if each graph $G_i$ does not contain a $\T_i$-split.
\end{definition}

In Chekuri and Xu's algorithm, they add a post-processing step turning a prime decomposition into a \emph{canonical decomposition} (see also~\cite{Cheng99, Cunningham83}), where every mincut of $G$ can be found in exactly one graph from the canonical decomposition.
Indeed, 
it is possible for some $\T$-mincuts of $G$ not being preserved anymore in a prime decomposition.
For example, if the algorithm decomposes a graph using a $\T$-split, then all the $\T$-mincuts that cross with that split no longer exist in the simple refinement.
The only guarantee to any decomposition $\G=\{(G_i, \T_i)\}$ is 
that every $\T_i$-mincut in $G_i$ corresponds to some $\T$-mincut in $G$ (by ``expanding'' the anchor vertices with the terminal vertices in $\T$).
Fortunately, all $\T$-mincuts that the algorithm has missed in one divide and conquer step belong to the same cycle on a cactus representation of $\T$.
In \Cref{sec:compute-cactus-from-prime-decomposition-tree} we show that even without the post-processing step,
we are still able to efficiently \emph{glue} the cactus of decomposed graphs (via anchor vertices) such that a cycle on a cactus representation can still be constructed.
Thus, all $\T$-mincuts are preserved.

However, iteratively finding splits and repeatedly invoking simple refinements have a worst case $\Omega(m|\T|)$ runtime, which is too slow.
This worst case occurs when the splits used for simple refinements were imbalanced.
In the rest of the section, we resolve this issue via maximal isolating mincuts, obtaining an algorithm for Steiner cactus in poly-logarithmic max-flow time.

\subsection{Our Divide and Conquer Framework via a Sequence of Splits}
\label{sec:divide and conquer framework}

Our divide and conquer algorithm is based on the idea of Chekuri and Xu~\cite{CX17}, where the goal is to
output a prime decomposition through a series of simple refinements.
However, in each subproblem, instead of seeking one split at a time,
our algorithm uses multiple splits in $G$ and generates a \emph{good} decomposition (see \Cref{def:good-decomposition}) with high probability.
The guarantee of a good decomposition leads to an $O(\log|\T|)$ upper bound to the recursion depth, achieving a poly-logarithmic max-flow runtime.

In the rest of this subsection, we formulate a divide and conquer framework (see \Cref{alg:divide-and-conquer-framework}) for computing a $\T$-Steiner cactus.
The implementation of this framework has to overcome two non-trivial challenges, namely (1) computing a collection of splits that generates a good decomposition (\Cref{lem:compute good splits}) and (2) merging the sub-cactus returned from the subproblems into a $\T$-Steiner cactus (\Cref{lem:merge-cactus}).
We overcome both challenges using our maximal isolating mincut algorithm and describe the details in \Cref{sec:compute good splits collection} and \Cref{sec:compute-cactus-from-prime-decomposition-tree}.
By assuming \Cref{lem:compute good splits,lem:merge-cactus}, we establish a proof to~\Cref{thm:steiner-cactus-main} at the end of this subsection.

\paragraph{Preprocessing.}
To enable the power of the maximal isolating mincut algorithm, we rely on the following two handy properties after preprocessing:
\begin{enumerate}[itemsep=0pt]
\item We may assume that we have already known the value of the Steiner mincut $\lambda_G(\T)$.
\item We may assume that for any two terminal vertices $u$ and $v\in \T$, there exists a $\T$-Steiner mincut that separates $u$ and $v$.\label{assumption:2}
\end{enumerate}
These two assumptions can both be achieved using the isolating cut algorithm from Li and Panigrahi~\cite{LP20}.
The first assumption can be made directly via an almost-linear time algorithm~\cite{CQ21}
that computes a $\T$-mincut and its value.
The second assumption can be made by preprocessing the graph with a ``$\lambda$-connected component algorithm'' implicitly mentioned in Li and Panigrahi~\cite{li2021approximate}.
We summarize the second preprocessing step below in \Cref{lem:preprocessing} and prove them in \Cref{sec:proof-of-preprocessing} for completeness.

\begin{lemma}[Preprocessing~{\cite[run one step of Algorithm 4]{li2021approximate}}]\label{lem:preprocessing}
Given a graph $G$ and a terminal set $\T$, there exists an algorithm such that, with probability $1-n^{-11}$ the algorithm outputs a partition of $\T$ such that $\lambda(u, v) = \lambda_G(\T)$ if and only if $u$ and $v$ belongs to different parts. This algorithm runs in $O(\log ^2 n\cdot\MaxFlow(2n, 2m))$ time.
\end{lemma}

\paragraph{Good Decomposition.}
Let us now consider decomposing the graph $G$ using one or more splits at a time in Chekuri and Xu's algorithm.
Suppose we have an ideal oracle that always finds a \emph{balanced $\T'$-split} where both sides contain at least $\frac14|\T'|$ terminal vertices in a graph $G'$ with terminal vertex set $\T'$.
Then, performing one simple refinement in each subproblem $(G', \T')$ suffices to bound the recursion depth by $O(\log|\T|)$.
Unfortunately, such an exemplary oracle does not always exist.
In an extreme scenario, consider a graph $G$ and a terminal set $\T$ whose Steiner cactus representation could be a star.
Since all $\T$-mincuts are trivial, there is even no split for $\T$.
Fortunately, if there is no further split for $\T$ on $G$, then $\{(G, \T)\}$ itself is already a prime decomposition so this is a base case in our divide and conquer algorithm.
Motivated by this, we define the \emph{good decomposition} that suffices to bound the recursion depth as follows.

\begin{definition}[Good Decomposition]\label{def:good-decomposition}
    Given a graph $G$ and a set of terminal vertices $\T$,
    a decomposition $\mathcal{G}=\{(G_i, \T_i)\}$ of $G$ is said to be \emph{good} with respect to $\T$ if $\mathcal{G}$ has the following property. Let $\T_i$ be the set of terminal vertices in $G_i$.
    For all $i$ except at most one special index $i^*$, $|\T_i| \le \frac34 |\T|+1$, and there exists a Steiner cactus representation of $\T_{i^*}$ in $G_{i^*}$ that is a star.
\end{definition}

\paragraph{Induced Decomposition from a Collection of Disjoint Splits.} Consider a collection of $\T$-splits $\mathcal{S}=\{X_1, X_2, \ldots, X_k\}$ where the presented subsets in $\mathcal{S}$ are \emph{disjoint}.
The disjointness leads to a robust procedure for performing lots of simple refinements to $G$ using the $\T$-splits from $\mathcal{S}$ in any order.
It is straightforward to check that the resulting decomposition is unique up to relabeling the anchor vertices, and we call the result \emph{the decomposition induced by} $\mathcal{S}$ on $G$.

Now, we formally establish sufficient criteria that our maximal isolating mincut algorithm will achieve with high probability.

\begin{definition}[Good Split Collection]\label{def:good-splits}
    Given a graph $G$ and a set of terminals $\T$, we say that a collection of $\T$-splits $\mathcal{S}=\{X_i\}$ is a \emph{good split collection} if 
    (1) for any $i\neq j$ we have 
    $X_i\cap X_j=\emptyset$, and
    (2) the decomposition induced by $\mathcal{S}$ is a good decomposition.
\end{definition}

With the above \Cref{def:good-splits}, we are able to summarize and highlight the first step
in the divide and conquer framework in \Cref{lem:compute good splits} (proved in \Cref{sec:compute good splits collection}).
The entire divide and conquer algorithm is presented in \Cref{alg:divide-and-conquer-framework}.

\begin{lemma}
\label{lem:compute good splits}
    Given a graph $G = (V, E)$ and a set of terminals $\T$, there exists a randomized Monte Carlo algorithm such that, with probability $1-n^{-11}$, the algorithm returns a good split collection $\mathcal{S}$ in $O((\log^3 n)\cdot \MaxFlow(3n, 4m))$ time.
\end{lemma}

\begin{algorithm}[h]
\caption{A divide and conquer framework that computes a $\T$-Steiner cactus.}
\label{alg:divide-and-conquer-framework}
\begin{algorithmic}[1]
\Require{a graph $G=(V, E)$, a terminal set $\T$.}
\Ensure{a $\T$-Steiner cactus of $\G$.}
\Procedure{ComputeSteinerCactus}{$G, \T$}
\If{$|\T| \le 3$}\Comment{Either a triangle or a path.}
\State \Return \Call{TrivialCactus}{$G, \T$}.\label{line:merge-cuts-1}
\Else
\State Obtain $\mathcal{S}$, a good collection  of $\T$-splits on $G$.\Comment{See~\Cref{lem:compute good splits} and \Cref{sec:compute good splits collection}.}\label{line:compute-good-collection}
\If{$\mathcal{S}=\emptyset$} 
\State \Return \Call{StarCactus}{$G, \T$}.\footnotemark \Comment{There will be two types of stars, see \Cref{sec:compute-cactus-from-prime-decomposition-tree}.}\label{line:merge-cuts-2} 
\EndIf
\State Compute the decomposition $\mathcal{G}=\{(G_i, \T_i)\}$ induced by $\mathcal{S}$ over $G$.
\label{line:compute-decomposition}
\State
Obtain $H_i\gets \textsc{ComputeSteinerCactus}(G_i, \T_i)$ for all $i$.
\State
\Return \Call{MergeCactus}{$G, \T, \{H_i\}$}. \Comment{See \Cref{lem:merge-cactus} and~\Cref{sec:compute-cactus-from-prime-decomposition-tree}.}\label{line:merge-cuts-3}
\label{line:return-decomposition}
\EndIf
\EndProcedure
\end{algorithmic}
\end{algorithm}
\footnotetext{In the hypergraph setting, we replace this procedure with \textsc{StarOrBrittleCactus}, see~\Cref{sec:returning-a-correct-cactus-hyergraph}.}

\paragraph{Merging Cactus from Subproblems.}
The last piece for accomplishing the divide and conquer algorithm is to merge the cactus returned from each subproblem.
On the bright side, with the help of anchor vertices, we do have the proximity of how two cactus should be combined.
However, the merging procedure is a bit subtle as
we have to make sure that every $\T$-Steiner mincut 
is preserved in the combined cactus.
We summarize the correctness and the runtime guarantee here in \Cref{lem:merge-cactus} and establish the details in \Cref{sec:compute-cactus-from-prime-decomposition-tree}.

\begin{lemma}\label{lem:merge-cactus}
Fix a subproblem $(G=(V, E), \T)$ in \Cref{alg:divide-and-conquer-framework}.
Assume all splits generated from the subproblems $\{(G_i, \T_i)\}$ derived from $(G, \T)$ are good, and each subproblem returns a correct $\T_i$-Steiner cactus of $G_i$.
Then, 
the procedures \textsc{TrivialCactus}, \textsc{StarCactus}, and \textsc{MergeCactus} returns a $\T$-Steiner cactus of $G$  in $O(\log|\T|\cdot \MaxFlow(2n, 2m))$ time.
\end{lemma}

Now, with \Cref{lem:compute good splits,lem:merge-cactus}, we are able to prove the main \Cref{thm:steiner-cactus-main} by 
completing the (relatively trivial) implementation details of Line~\ref{line:compute-decomposition} and
analyzing its runtime.

\begin{proof}[Proof of \Cref{thm:steiner-cactus-main}]
Let $\mathcal{S}=\{X_1, X_2, \ldots, X_\ell\}$ be a good split collection returned from Line~\ref{line:compute-good-collection}.
To obtain a decomposition $\G$ induced by $\mathcal{S}$ (Line~\ref{line:compute-decomposition}),
the algorithm first computes the induced subgraphs $G[X_i]$ for all $i$ in linear time.
Then, the algorithm simulates the simple refinement of $X_i$ by creating an anchor vertex $a_i$ for each split $X_i$, and
for each edge $(u, v)$ that across the split $u\in X_i$ but $v\notin X_i$, the algorithm either adds a new edge from $(u, a_i)$, or adds the weight to an existing edge $(u, a_i)$.
Finally, the algorithm duplicates the graph $G$ and contract each subset $X_i$ into a single anchor vertex $a_i$, forming the last decomposed graph $(G_{\ell+1}, \T_{\ell+1})$. The implementation of Line~\ref{line:compute-decomposition} takes linear time $O(|V(G)|+|E(G)|)$ in total.

\paragraph{Runtime.}
Consider the recursion tree of subproblems from \Cref{alg:divide-and-conquer-framework}.
By definition of a good decomposition, we know that the recursion depth satisfies the following recurrence relation: $\mathrm{MaxDepth}(k) = \mathrm{MaxDepth}(\lfloor{\frac34 k}\rfloor + 1) + 1$ whenever $k>3$ and $\mathrm{MaxDepth}(k) = 0$ whenever $k\le 3$. By solving the recurrence relation we obtain $\mathrm{MaxDepth}(k) = O(\log k)$.

Now, it suffices to bound the total subproblem sizes within the same recursion depth.
We first claim that the number of vertices that occur across all subproblems in the same recursion depth is at most $2n$ using a potential method.
For each subproblem $(G, \T)$ we define an invariant potential $\Phi(G, \T) := |V(G)|+|\T|-4$.
Notice that in the case where at least one recursion step is performed we must have $|\T|\ge 3$ and hence $\Phi(G, \T) > 0$. 
If $\{(G_1,\T_1),(G_2,\T_2)\}$ is a simple refinement of $G$, observe that \[\Phi(G_{1},\T_{1})+\Phi(G_{2},\T_{2})=(|V(G_{1})|+|V(G_{2})|)+(|\T_{1}|+|\T_{2}|)-8=(|V(G)|+2)+(|\T|+2)-8=\Phi(G,\T).\] Thus, consider the induced decomposition $\{(G_i, \T_i)\}$ on a good split collection of size $k$, we know that $\Phi(G, \T) = \sum_{i=1}^{k+1} \Phi(G_i, \T_i)$.
Therefore, the sum of all potentials within the same recursion depth can be upper bounded by the root problem's potential.
Since in every subproblem we have $|\T|\ge 3$, we conclude that the total number of vertices across all subproblems at any particular recursion depth (or any collection of subproblems that are not related to each other) is at most $\Phi(G, |\T|) = n+|\T|-4\le 2n$. 

As a consequence, we also deduce that there are at most $4|\T|-9$ subproblems in the recursion tree, by noticing that the recursion tree is a branching tree with at most $2|\T|-4$ leaf subproblems (every leaf subproblem contains at least one anchor vertex and every vertex in $V(G)\setminus \T$ occurs in exactly one leaf subproblem.)

To bound the total number of edges across all subproblems within the same recursion depth, we observe that after computing an induced decomposition from a good split collection, the total number of edges is increased by at most $\sum_{i=1}^\ell |V(G_i)|\le 2n$ (notice that we charge the number of the newly generated edges in the last decomposed graph $(G_{\ell+1}, \T_{\ell+1})$ to the edges across each split). Hence, we know that at any recursion depth there are at most $m+2n\log|\T|$ edges in total.

Finally, we add up the runtime needed per recursion depth.
Fix any recursion depth, for each subproblem $(G_j, \T_j)$, by~\Cref{lem:compute good splits}   the runtime spent in Line~\ref{line:compute-good-collection} is at most 
$O((\log^3 n)\cdot\MaxFlow(3|V(G_j)|, 4|E(G_j)|))$, the runtime spent for Line~\ref{line:compute-decomposition} is linear in the graph size $O(|V(G_j)|+|E(G_j)|)$,
and by \Cref{lem:merge-cactus} merging cactus takes $O(\log|\T_j|\cdot \MaxFlow(|V(G_j)|, |E(G_j)|))$ time. 
Hence, by denoting $k=|\T|$,
the runtime of \Cref{alg:divide-and-conquer-framework} is
\begin{align*}
    & \underbrace{O(\log k)}_{\text{recursion depth}}\ \cdot\ O(\underbrace{(m+2n\log k)}_{\text{Line~\ref{line:compute-decomposition}}} + \underbrace{(\log^3 n)\cdot \MaxFlow(3n, 4(m+2n\log k)))}_{\text{Line~\ref{line:merge-cuts-1}, Line~\ref{line:compute-good-collection}, Line~\ref{line:merge-cuts-2}, and Line~\ref{line:merge-cuts-3}}} \\
    =\ & O(m\log k+n\log^2 k + (\log^3 n)\cdot\MaxFlow(3n, 4m+8n\log k))\\
    =\ & O((\log^4 n)\cdot\MaxFlow(3n, 4m+8n\log k)).
\end{align*}
Finally, note that we spent time for preprocessing the input graph by called \Cref{lem:preprocessing} using  $O(\log ^2 n\cdot\MaxFlow(2n, 2m))$ time, but this is subsumed by the above bound.

\paragraph{Correctness.} By~\Cref{lem:compute good splits}, with probability $1-n^{-11}$ the returned collection is good in Line~\ref{line:compute-good-collection}.
Throughout execution there are at most $4|\T|-9\le 4n$ invocations to~\Cref{lem:compute good splits}.
Hence, with a union bound we know that with probability $1-4n^{-10}$ the collections of splits from all subproblems are good.
Now, by applying the union bound again to \Cref{lem:merge-cactus} we know that the returned cactus is a $\T$-Steiner cactus of $G$ with probability at least $1-4n^{-10}$.
Therefore, with another union bound we know that with probability $1-8n^{-10}$ \Cref{alg:divide-and-conquer-framework} correctly outputs a $\T$-Steiner cactus.
\end{proof}

\subsection{Computing a Good Split Collection}
\label{sec:compute good splits collection}

In this subsection, we aim to prove \Cref{lem:compute good splits}. 
Specifically, we propose \Cref{alg:compute good splits}, and then we show that with probability at least $1-n^{-11}$, a good split collection can be computed in almost-linear time via $O(\log^2n)$ maximal isolating mincut algorithms.

\paragraph{Algorithm Description.} 
\Cref{alg:compute good splits} works as follows.
The algorithm set up $\lceil\log |\T|\rceil$ different sample rates, namely $2^{-1}, 2^{-2}, \ldots, 2^{-\lceil \log|\T|\rceil}$. For each sample rate $2^{-i}$, the algorithm samples each terminal vertex with probability $2^{-i}$ and forms a set $\T_i$.
Then, the algorithm computes the maximal isolating mincuts for the set $\T_i$, and keeps the maximal $v$-isolating mincut of $\T_i$ if the cut is a $\T$-split and its value equals to $\lambda_G(\T)$ (i.e., keeps only the non-trivial $\T$-Steiner mincuts.)
To ensure a high probability result, we repeat the whole sampling procedure another $\Theta(\log n)$ times.
Let $\mathcal{S}$ be the collection of splits that the algorithm has found so far.

Recall from \Cref{def:good-splits} that there are two cases where a split collection is considered to be good:
either we find a balanced split whose both sides have at least $\frac14|\T|$ terminals, or we find a collection of disjoint sets where the contracted graph (obtained by contracting all these sets) does not contain a split anymore.

Once obtaining the collection of splits $\mathcal{S}$, the algorithm checks if there exists any balanced split by simply checking the size of each set in $\mathcal{S}$. If there is such a balanced split, returning the split itself is sufficient (Line~\ref{line:check-balanced-split}).
Otherwise, every set in $\mathcal{S}$ now contains either less than $\frac14|\T|$ or more than $\frac34|\T|$ terminals.
The algorithm discards all sets containing more than $\frac34|\T|$ terminals
and then keeps the maximal subsets among the splits in the collection.

In the case where no balanced split is found, the algorithm does an additional post-processing in Line~\ref{line13}-\ref{line15}. The purpose of this post-processing is to obtain a set of disjoint $\T$-splits that satisfy \Cref{def:good-splits}.
Specifically, in Line~\ref{line13} the algorithm get rid of all subsets with only one terminal. In Line~\ref{line14} only maximal subsets are kept and in Line~\ref{line15} the disjointness of these subsets are enforced.
This completes the description of \Cref{alg:compute good splits}.

\begin{algorithm}[h]
\caption{Computing a Good Split Collection}
\label{alg:compute good splits}
\begin{algorithmic}[1]
\Procedure{GoodSplitCollection}{$G, \T, \lambda := \lambda_G(\T)$}
\State Initialize $\mathcal{S} \gets \emptyset$.

\renewcommand{\algorithmicwhile}{\textbf{repeat}}
\While{the following procedure $\lceil 12\cdot 1024e\cdot\ln n \rceil$ times}
    \For{$i = 1, 2, \ldots, \lceil\log |\T|\rceil$}\label{line:good-split-collection-for-loop}
        \State Sample each terminal vertex with probability $2^{-i}$, denote the set by $\T_i$.
        \State $\mathcal{X}\gets $\Call{MaxIsoCut}{$G, \T_i$}.\label{line:alg-good-split call max iso}
        \State $\mathcal{S} \gets \mathcal{S}\cup \{X\in \mathcal{X} \ |\ \C(X)=\lambda\}$.\label{line7}
    \EndFor
\EndWhile \label{line9}
\renewcommand{\algorithmicwhile}{\textbf{while}}

\If{there exists a balanced split $X\in \mathcal{S}$ where $|\T \cap X|,|\T \setminus X|\ge \frac{1}{4}|\T|$}\label{line10}
    \State\Return $\{X\}$. \Comment{A single balanced split.}\label{line:check-balanced-split}
    \Else
    \State $\mathcal{S}\gets \{X_i\in \mathcal{S} \ |\ 2\le |X_i\cap \T|\le \frac14|\T|\}$. \Comment{Keep only small $\T$-splits.}\label{line:keep-only-small-t-splits}\label{line13}
    \State $\mathcal{S}\gets \{X_i\in \mathcal{S}\ |\ X_i\not\subseteq X_j \text{ for all $i\neq j$}\}$.\Comment{Obtain only the maximal subsets.}\label{line14}
    \State For each $X_i\in\mathcal{S}$, set $X_i=X_i\setminus \cup_{j<i} X_j$.\Comment{Enforce disjointness to the subsets.}\label{line:enfore-disjoint}\label{line15}
    \State\Return $\mathcal{S}$. \label{line:return-maximal-subsets}
    \EndIf
\EndProcedure
\end{algorithmic}
\end{algorithm}

To simplify the correctness proof, we introduce the notion of \emph{irredundant} $\T$-Steiner cactus.

\begin{definition}\label{def:irredundant-cactus}
A $\T$-Steiner cactus $(H, \phi)$ of $G$  is said to be \emph{irredundant}, if for every edge $e$ on $H$, the contraction $H/e$ is no longer a $\T$-Steiner cactus of $G$.
\end{definition}

We remark that the irredundant cactus is somewhat similar to the notion of a normal cactus defined in the work of Nagamochi and Kameda~\cite{nagamochi1994canonical}, except that we still allow tree edges in $H$.

\paragraph{Intuition of Correctness.}
The high probability correctness
comes from case analysis to \emph{any} $\T$-Steiner cactus of $G$.
Let $(H, \phi)$ be any irredundant $\T$-Steiner cactus of $G$.
Define a \emph{balanced edge-cut} on $H$ to be a minimum edge cut of $H$ (either one edge or two edges in a cycle) such that the number of terminals on both sides is between $\frac14|\T|$ and $\frac34|\T|$.
Our analysis depends on whether or not a balanced edge-cut exists on $H$.\footnote{The existence of a balanced edge-cut on $H$ is equivalent to the existence of a balanced split on $G$. However, we believe the proof is easier to see through if we analyze the algorithm's behavior on $H$.}

\paragraph{Case 1: Balanced Cuts Exist.} In the first case where there is a balanced edge-cut, the correctness relies on the sparsest sampling rate $2^{-\lceil\log|\T|\rceil}$.
In particular, we rely on a sampled terminal set $\T'=\{u, v, r\}$ of exactly 3 vertices, where two corresponding nodes $\phi(u)$ and $\phi(v)$ are in the ``larger side'' of the cut and the third corresponding node $\phi(r)$ is in the ``smaller side'' of the cut.
Then, it is possible to prove that with constant probability,  the maximal $r$-isolating cut of $\T'$ contains the right amount of terminal vertices --- between $\frac14|\T|$ and $\frac34|\T|$ (the upper bound comes from~\Cref{lem:sample-two-vertices}).
Thus, a balanced split will be found in $\mathcal{S}$ with high probability because the sampling procedure with the sparsest sampling rate is repeated $O(\log n)$ times.
We formalize the first case here as \Cref{lem:steiner-mincut-case-1}, and give an illustration in 
\Cref{fig:case1}.

\begin{figure}[htbp]
    \centering
    \includegraphics[width=0.9\textwidth]{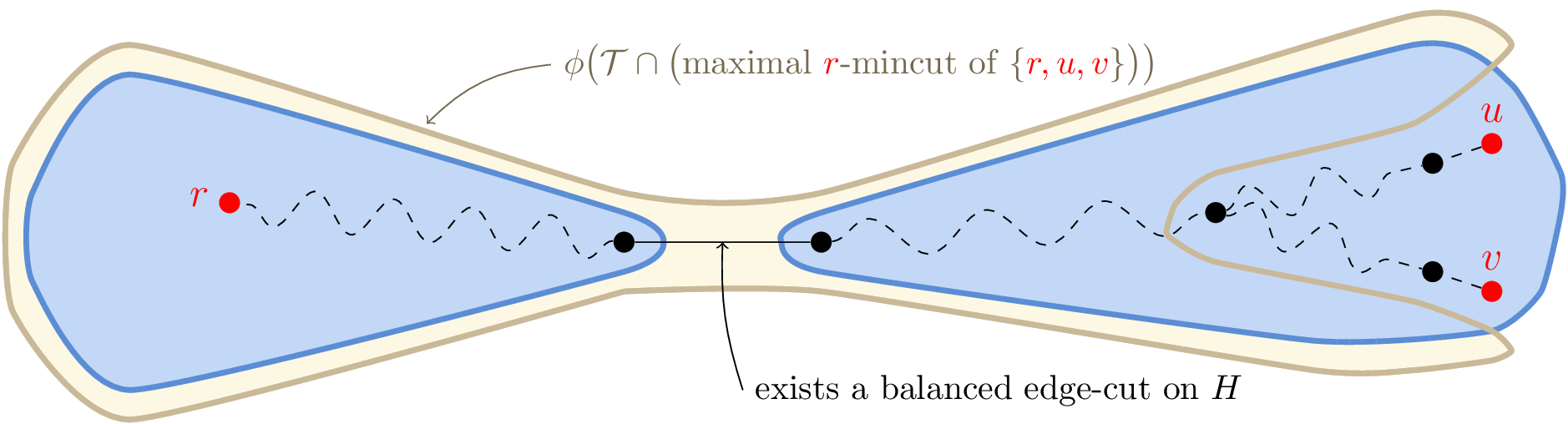}
    \caption{When there is a balanced edge-cut on a $\T$-Steiner cactus $H$ of $G$, a balanced $\T$-split on $G$ will be found with high probability.}
    \label{fig:case1}
\end{figure}

\begin{lemma}\label{lem:steiner-mincut-case-1}
Let $G$ be the graph with terminal set $\T$ and let $(H, \phi)$ be a $\T$-Steiner cactus of $G$.
Suppose there is a balanced edge-cut on $H$.
Then, with probability
$1-n^{-11}$
there is a balanced split in $\mathcal{S}$ returned from \Cref{alg:compute good splits}.
\end{lemma}

To prove \Cref{lem:steiner-mincut-case-1}, we introduce a helper lemma that shows the benefit of having assumption~\ref{assumption:2}. That is, if we have sampled two vertices $u$ and $v$ in the ``large side'' and sampled a single vertex $r$ in the ``smaller side'' of a balanced split,
then the maximal $r$-isolating cut of $\{u, v, r\}$ has to contain \emph{at most} $\frac{1}{2}|\T|$ terminals.

\begin{lemma}\label{lem:sample-two-vertices}
Let $G$ be the graph with a set $\T$ of terminals that satisfies Assumption~\ref{assumption:2}.
Let $r\in\T$ be a terminal such that any $r$-isolating mincut is a $\T$-mincut.
If we sample terminals $u, v\in \T-\{r\}$ uniformly at random,
then with probability at least $1/4$, 
any $\{u, v\}$-mincut of $\{u, v, r\}$ has at least
$\frac12|\T|$ terminals.
\end{lemma}

\begin{proof}
Let $(H, \phi)$ be a $\T$-Steiner cactus of $G$. 
Recall that Assumption~\ref{assumption:2} states that any two terminals $u, v\in \T$ can be separated by some $\T$-mincut.
This implies that $\phi(u)\neq \phi(v)$ as all $\T$-mincuts are preserved.
From the assumption that $r$-isolating mincut is a $\T$-mincut, we know that $\phi(r)$ has degree 1 or has degree 2 within a cycle in $H$.
Consider a specialized DFS traversal of $H$ starting from $\phi(r)$.
Upon visiting a vertex from a cycle edge, the DFS traversal always tends to choose any edge that leaves the cycle.
Let $(r, v_1, v_2, \ldots, v_{|\T|-1})$ be the unique permutation of $\T$ where $(\phi(r), \phi(v_1), \phi(v_2), \ldots, \phi(v_{|\T|-1}))$ is the order (subsequence) of visited vertices by the DFS traversal, i.e. the pre-order.
Notice that the DFS traversal only returns to $\phi(r)$ at the very end.
Then for any two indices $i$ and $j$ such that $1\le i < j\le |\T|$, maximal $r$-isolating mincut of $\{v_i, v_j, r\}$ must not contain any vertices in $\{v_i, v_{i+1}, \ldots, v_j\}$ and hence the result follows by counting the fraction of pairs (at least $1/4$) whose position in the permutation differs by at least $\frac12|\T|$.
\end{proof}

\Cref{lem:sample-two-vertices} implies that with constant probability, the mincut of our concern is balanced.
Now we are ready to prove~\Cref{lem:steiner-mincut-case-1}.

\begin{proof}[Proof of~\Cref{lem:steiner-mincut-case-1}]
Suppose there is a balanced edge-cut on $H$, i.e. a $\T$-Steiner mincut $(X, V\setminus X)$ where $\frac14|\T|\le |X\cap \T| \le \frac12|\T|$.
Consider the phase that \Cref{alg:compute good splits} samples each terminal vertex with probability $2^{-\lceil\log |\T|\rceil}$, then with probability at least 
\begin{align*}
& {\frac12|\T| \choose 2}\cdot\frac14|\T| \left(1-\frac1{2^{\lceil \log|\T|\rceil}}\right)^{|\T|-3}\cdot\left(\frac{1}{2^{\lceil \log|\T|\rceil}}\right)^3 \\
& \ge ~
\frac{|\T|^3}{32e}\left(1-\frac{2}{|\T|}\right)\left(1-\frac{1}{2(|\T|-1)}\right)^{-3} \left(\frac{1}{2(|\T|-1)}\right)^{3} \\
& \ge ~ \frac{1}{256e} ~.
\end{align*}
there will be one sampled terminal $r$ in $X$, and exactly two sampled terminals $u$ and $v$ in $V\setminus X$.
We denote the event described above by $\mathcal{E}$. When $\mathcal{E}$ happens,
since $(X, V\setminus X)$ is a $\T$-Steiner mincut, we know that the maximal $r$-mincut of $\{r, u, v\}$ must contain entire $X$, hence containing at least $\frac14|\T|$ terminals.
Now
we will use~\Cref{lem:sample-two-vertices} to prove that, conditioned on $\mathcal{E}$,
with probability at least $1/4$, the maximal $r$-mincut of $\{r, u, v\}$ has at most $\frac34|\T|$ terminals.

Indeed, conditioned on event $\mathcal{E}$, 
the maximal $r$-mincut of $\{r, u, v\}$ is disjoint to the minimal $\{u, v\}$-mincut of $\{r, u, v\}$ (otherwise it contradicts to Disjoint \& Posi-modularity~\Cref{lem:disjoint-posi-modularity}).
Consider the graph $G/X$ with $r'$ being the contracted terminal vertex.
Let $\T':=(\T\setminus X)\cup \{r'\}$ be the contracted terminal set.
Since $X$ itself is a $\T$-Steiner mincut, $\{r'\}$ is a $\T'$-Steiner mincut on $G/X$ and hence the criteria of~\Cref{lem:sample-two-vertices} are met.
Therefore, by~\Cref{lem:sample-two-vertices}, with probability at least $1/4$, the minimal $\{u, v\}$-mincut of $\{r, u, v\}$ on $G$ contains at least $\frac12|\T'|\ge \frac14|\T|$ terminals. This implies that the maximal $r$-mincut of $\{r, u, v\}$ contains at most $\frac34|\T|$ terminals.

As a consequence, we know that with probability $1/(1024e)$, a balanced split will be found in one sampling procedure. With repeating the sampling procedure for $\lceil 12\cdot 1024e\cdot \ln n \rceil$ times, \Cref{alg:compute good splits} returns a balanced split with probability at least $1-n^{-12} \ge 1-n^{-11}$ as desired.
\end{proof}

\paragraph{Case 2: No Balanced Cut.} Now let us consider the second case where there is no balanced edge-cut.
An illustration for this case is provided in \Cref{fig:case2}.
In this case, our algorithm should be able to obtain lots of disjoint $\T$-splits such that,
after contracting smaller sides of all these $\T$-splits,
the remaining graph (which could still be large) is guaranteed to have a star shaped cactus representation.
Not surprisingly, the center of this star can be traced back (by undoing the contractions) to a \emph{centroid} node on $H$, given the non-presence of a balanced edge-cut.

Recall that $H$ is an irredundant $\T$-Steiner cactus representation $(H, \phi)$ of $G$.
A \emph{centroid} $v$ is a node on $H$ such that, every 
edge or cycle incident to $v$ defines a $\T$-Steiner mincut
whose corresponding mincut on $H$ not containing $v$    has at most $\frac14|\T|$ terminals.
We will soon prove
(in~\Cref{lem:steiner-mincut-case-2-part1}) that no balanced edge-cut on $H$ implies a unique centroid node $v$ on $H$.
This centroid node $v$ naturally partitions $\T\setminus \phi^{-1}(v)$ into sets of terminals $T_1\sqcup T_2\sqcup \cdots \sqcup T_k$, where for each $i$, $\phi(T_i)$ belongs to the same connected component in $H-v$.
Moreover, since $(H, \phi)$ is a cactus representation for $G$, we know that for each $T_i$ there exists a $\T$-Steiner mincut that separates $T_i$ and $\T\setminus T_i$.
Let $X'_i$ be the maximal $T_i$-mincut of $\T$, and define the collection
$\mathcal{S}'=\{X'_i\}$.

Fix a particular $i$ such that $1\le i\le k$, and consider sampling each vertex in $\T$ at the sampling rate $2^{-\lceil{\log |T_i|}\rceil}$.
We can then prove (in~\Cref{lem:steiner-mincut-case-2-part2})
that,
with constant probability, \emph{exactly one} terminal $w \in T_i$ is sampled, together with \emph{at least one} terminal from any two other subsets (namely $x\in T_j$ and $y\in T_{j'}$, for some $j\neq j'\neq i\neq j$) being sampled.

The following lemma ensures that $X'_i$ can be precisely discovered by our maximal isolating mincut algorithm, illustrated in \Cref{fig:case2}(b).

\begin{lemma}\label{lem:avoid-centroid}
    Let $v$ be a centroid node on $H$.
    Let $w, x, y\in \T$ be three terminals such that $\phi(w), \phi(x), $ and $\phi(y)$ belongs to distinct connected components in $H-v$.
    Let $X$ be the maximal $w$-isolating mincut of $\{w, x, y\}$, then the corresponding cut of $X$ in $H$ must not contain $v$. 
\end{lemma}

\begin{proof}
Let $Y$ be a mincut on $H$ that corresponds to $X$.
    Suppose by contradiction that $Y$ contains $v$.
    But since $\phi(x), \phi(y)\notin Y$, the cut value of $Y$ would be at least $2\lambda_G(\T)$, a contradiction to $Y$ being a mincut of $H$.
\end{proof}

Notice that by \Cref{lem:avoid-centroid}, whenever the algorithm seeks the maximal $w$-isolating mincut of this sampled terminal set, the algorithm obtains exactly the set $X'_i$.
Is the collection $\mathcal{S}'$ serves for our purpose? Not really --- $\mathcal{S}'$ may not be a good split collection (\Cref{def:good-splits}). For example, some mincut $X'_i\in\mathcal{S}'$ may contain exactly one terminal vertex --- simply removing these mincuts is an easy fix.
What's worse, there could be two mincuts $X'_i$ and $X'_j$ that are not disjoint (e.g.,~\Cref{fig:case2}(a)).
Fortunately, an additional post-processing step can be further applied: whenever there exists $X'_i\cap X'_j\neq\emptyset$ (say $i > j$), we \emph{prune} the larger indexed one by replacing $X'_i$ with $X'_i\setminus X'_j$.
By posi-modularity (\Cref{lem:disjoint-posi-modularity}), $T_i\cap T_j=\emptyset$ implies that $X'_j$ is still a $T_j$-mincut of $\T$.
It is straightforward to check that at the end of the post-processing step we have obtained a pruned set $\mathcal{S}'_{\text{pruned}}=\{X_i\}$ where $X_i=X'_i\setminus \cup_{j<i} X'_j$ and every $X_i$ contains at least two terminal vertices.

The post-processing steps mentioned above correspond to Line~\ref{line15} of \Cref{alg:compute good splits}. 
We prove as a corollary of~\Cref{lem:steiner-mincut-case-2-part2} (\Cref{lem:transition-part2-to-part3}) that 
the collection $\mathcal{S}$ returned by~\Cref{alg:compute good splits} is exactly the same as $\mathcal{S}'_{\text{pruned}}$ with high probability.
At the end of the analysis, we prove in~\Cref{lemma:S-is-good-split-collection} that $\mathcal{S}'_{\text{pruned}}$ is actually a good split collection.

\begin{figure}[htbp]
    \centering
    \includegraphics[width=0.45\textwidth]{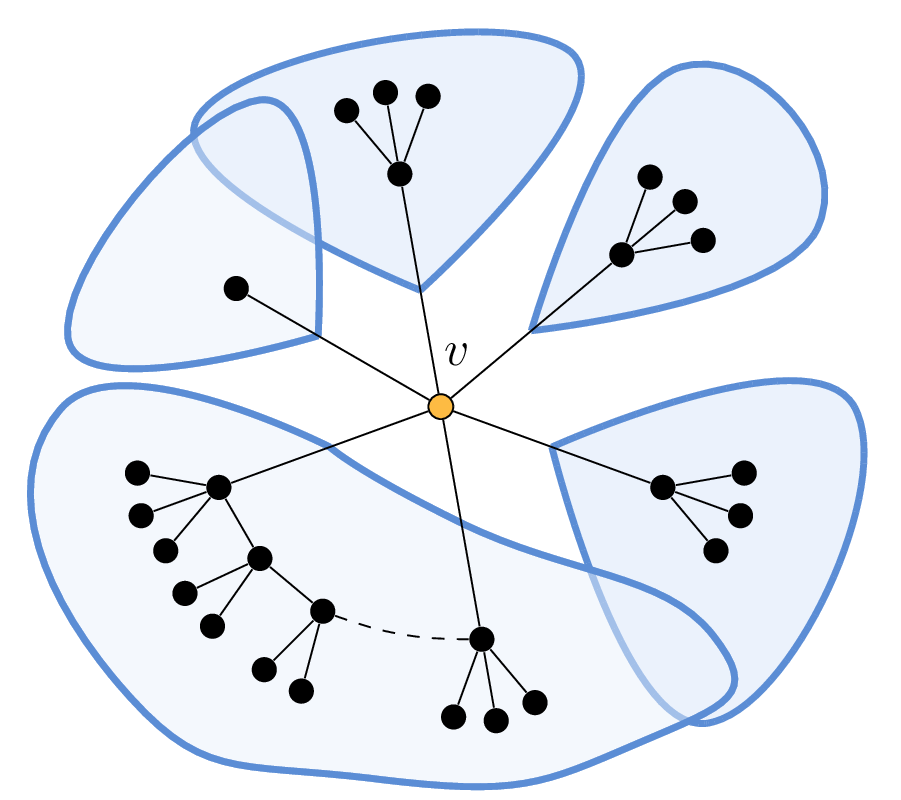}~
    \includegraphics[width=0.45\textwidth]{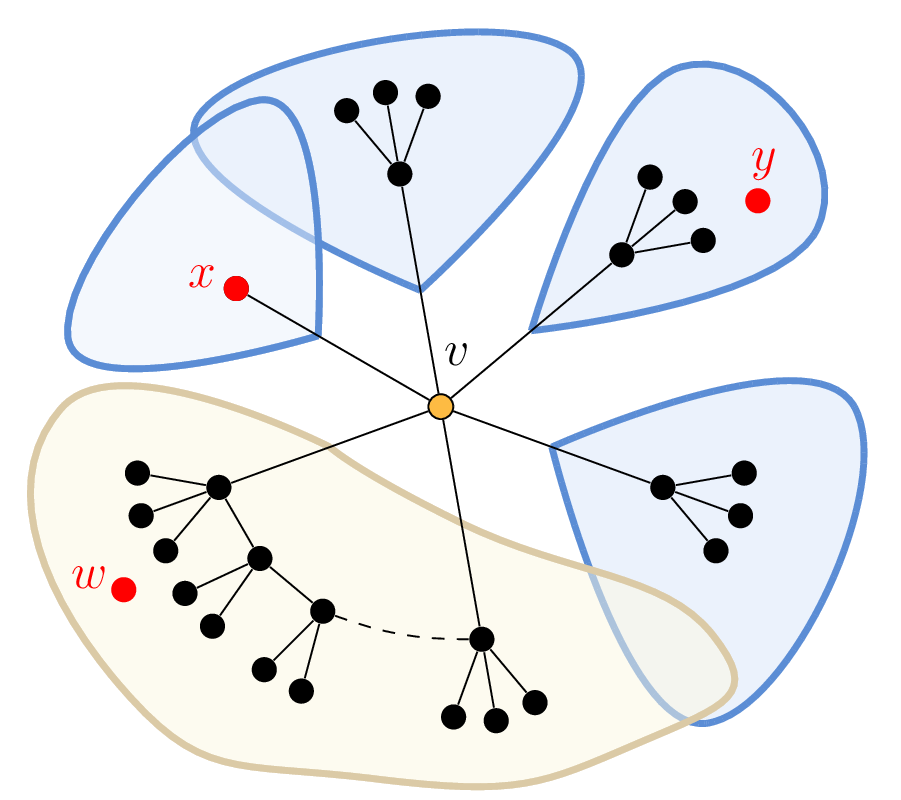}
    \caption{An illustration of a cactus representation $H$. The colored regions correspond to $\mathcal{S}'$, which is a collection of $\T$-Steiner mincuts of $G$.
    (a) When there is no balanced edge-cut, there must be a centroid node $v$ on $H$. (b) If we sample three terminals $w, x, y$ that are mapped into distinct connected components in $H-v$, then the maximal $w$-mincut of $\{w, x, y\}$ correspond to exactly the connected component of $H-v$ where $w$ belongs to.}
    \label{fig:case2}
\end{figure}

\paragraph{Formalizing the Proof to Case 2.} 
The rest of this subsection devotes to 
formalize the high-level idea described above. 
Let $G$ be the graph with terminal set $\T$ and let $(H, \phi)$ be an irredundant $\T$-Steiner cactus of $G$.
Assume that there is no balanced edge-cut on $H$.
We first show that there exists a unique centroid node on $H$.

\begin{lemma}\label{lem:steiner-mincut-case-2-part1}
there exists a unique centroid node $v$ on $H$ whose all incident 1-edges and 2-edges from the same cycle correspond to $\T$-Steiner mincuts of at most $\frac14|\T|$ terminals in the side not containing $v$.
\end{lemma}

\begin{proof}
    \label{proof:4.16-part1} The existence of such a centroid node can be proved as follows. We first replace each cactus cycle on $H$ with a star (adding an additional node that represents the cactus cycle), forming a tree $H'$.
We note that each edge on $H'$ still represents a mincut on $H$.
It is well-known that any tree contains a centroid with respect to any vertex weight.
Let $v$ be any centroid node on the tree $H'$ when empty nodes have weight zero and non-empty nodes have weight one.
If $v$ is a node on $H$, then by definition $v$ is a centroid node on $H$.
Otherwise, if $v$ is an additional node that represents a cactus cycle, then since each incident edge of $v$ on the tree $H'$ corresponds to a mincut with at most $\frac14|\T|$ terminals, one can obtain a balanced edge-cut on $H$ greedily along the cactus cycle represented by $v$, thereby a contradiction.

    To show uniqueness, suppose by contradiction that there are two centroids $u$ and $v$ on $H$. 
    Consider any mincut on $H$ that separates $u$ and $v$.
    By the fact of being a centroid, the side of this mincut not containing $u$ (resp. not containing $v$) has at most $\frac14|\T|$ terminals.
    Howerver, this implies that the total number of terminals is at most $\frac12|\T|$, a contradiction.
\end{proof}

Let $T_1\sqcup T_2\sqcup \ldots\sqcup T_k$ be the partition of terminal vertices of $\T\setminus \phi^{-1}(v)$ where two terminal vertices $x$ and $y$ belong to the same $T_i$ if and only if $\phi(x)$ and $\phi(y)$ are in the same connected component of $H-v$.
For each $i$, let $X'_i$ be the maximal $T_i$-mincut of $\T$ on $G$, and let $\mathcal{S}'=\{X'_i\}$ be the collection of all these mincuts.
We first establish the relation between the collection $\mathcal{S}'$  and the mincuts computed from the algorithm:

\begin{lemma}\label{lem:steiner-mincut-case-2-part2}
Let $\mathcal{\hat{S}}$ be the collection of $\T$-Steiner mincut right after the execution of Line~\ref{line14}.
With probability at least $1-n^{-11}$,
\begin{enumerate}[itemsep=0pt]
\item for all $X'_i\in \mathcal{S}'$ such that $|X'_i\cap \T|\ge 2$, we have $X'_i\in \mathcal{\hat{S}}$.
\item $\mathcal{\hat{S}}\subseteq \mathcal{S}'$.
\end{enumerate}
\end{lemma}

We first establish the following fact and a helper claim:
\begin{fact}\label{fact-for-case2}
Let $X$ be any $\T$-Steiner mincut.
Consider any corresponding edge-cut that separates $\phi(X\cap \T)$ and $\phi(\T\setminus X)$ on $H$.
Suppose that the centroid $v$ belongs to the $\phi(X\cap \T)$ side of the edge-cut, then $|X\cap \T| > \frac14|\T|$.
\end{fact}

\begin{proof}
Since $(H, \phi)$ is a cactus representation of $G$,
any minimum edge-cut of $H$ must have one side whose nodes are all  within the same connected component of $H-v$\footnote{This sentence even holds when $v$ is not a centroid.}.
The condition that $v$ belongs to the $\phi(X\cap\T)$ side implies that
all nodes in $\phi(\T\setminus X)$ belong to the same connected component of $H-v$.
Now, using the assumption that $v$ is a centroid node, we know that $|\T\setminus X| \le \frac14|\T|$ and hence $|X\cap \T| \ge \frac34|\T| > \frac14|\T|$.
\end{proof}

\begin{proposition}\label{prop-for-case-2}
    Fix any $X'_i\in \mathcal{S}'$, the maximal $T_i$-mincut of $\T$.
    For any $\T$-Steiner mincut $X\not\subseteq X'_i$ but $X\cap T_i\neq\emptyset$, we must have $|X\cap \T| > \frac14|\T|$.
\end{proposition}

\begin{proof}[Proof of \Cref{prop-for-case-2}]
Let $T_X=X\cap \T$ be the set of terminal vertices in $X$.
Consider the corresponding edge-cut on $H$ that separates $\phi(T_X)$ and $\phi(\T\setminus T_X)$. The condition of the proposition implies that $v$ belongs to the $\phi(T_X)$ side of the edge-cut.
Using \Cref{fact-for-case2} we obtain $|T_X| > \frac14|\T|$.
\end{proof}

Intuitively speaking, \Cref{prop-for-case-2} validates Line~\ref{line13}-\ref{line15} of \Cref{alg:compute good splits}: any $\T$-Steiner mincut that is either crossing or containing some $X'_i\in \mathcal{S}'$ will contain too many terminals and will be removed. As a consequence, Line~\ref{line13} protects $X'_i$ from being accidentally removed in Line~\ref{line14} and being ``chopped'' in Line~\ref{line15}.
Now we formally prove \Cref{lem:steiner-mincut-case-2-part2}.

\begin{proof}[Proof of \Cref{lem:steiner-mincut-case-2-part2}]~
\paragraph{Part 1.} Fix an $X'_i\in\mathcal{S}'$.
We first notice that by \Cref{prop-for-case-2},
any $\T$-Steiner mincut $X$ that is a superset of $X'_i$ has 
more than $\frac14|\T|$ terminal vertices. Hence, such mincut $X$ will be excluded by executing Line~\ref{line13} of the algorithm.
Now it suffices to show that $X'_i$ will be found and appeared in $\mathcal{\hat{S}}$ with high probability.

Consider sampling terminal vertices with probability $p:={2^{-\lceil\log |T_i|\rceil}}$.
By~\Cref{lem:avoid-centroid}, it suffices to lower bound the probability of the event where (1) no other terminals in $T_i$ is sampled and (2) some terminals are sampled from at least two other sets.
The following analysis further restricts condition (2): we first form a partition $\T\setminus \phi^{-1}(v)=T_i\sqcup Q_1\sqcup Q_2$ and compute the probability that at least one terminal from each of $Q_1$ and $Q_2$ are sampled.

Indeed, since each part has a size at most $\frac14|\T|$, by a straightforward greedy algorithm, it is possible to group all the parts except $T_i$ into two large sets, whose sizes are between $\frac14|\T|$ and $\frac58|\T|$\footnote{The greedy strategy iteratively merges each subset to the smaller pile of two. In the end, the difference between the two piles' sizes is at most $\frac14|\T|$. Since the sum of the two sizes is at least $\frac34|\T|$, the smaller pile must have at least $\frac14|\T|$ elements.}.
Let $Q_1$ and $Q_2$ be such two sets. We have $\T\setminus \phi^{-1}(v) = T_i\sqcup Q_1\sqcup Q_2$.
Now it suffices to prove that with constant probability, exactly one terminal is sampled from $T_i$ and at least one terminal is sampled from each of $Q_1$ and $Q_2$.
By a standard probability argument, we know that sampling exactly one terminal from $T_i$ has a probability of at least $1/(2e)$. For each $j\in\{1,2\}$, sampling at least one terminal from $Q_j$ has a probability at least:
\begin{align*}
    1-\left(1-p\right)^{|Q_j|} \ge 1-\left(1-\frac{1}{2|T_i|}\right)^{|Q_j|} &\ge 1-e^{-1/2}. \tag{$|Q_j|\ge \frac14|\T|\ge |T_i|$}
\end{align*}

Hence, the success probability per sampling at the particular scale is at least 
\begin{align*}
    \frac{1}{2e}\left(1-e^{-1/2}\right)^2 \ge 0.02  
\end{align*}
By repeating the sampling procedure at least $\lceil12\cdot (1/0.02)\cdot \ln n\rceil$ times, we know that \Cref{alg:compute good splits} obtains a $\T$-Steiner mincut that separates $T_i$ and $\T\setminus T_i$ with probability $1-n^{-12}$.
By applying another union bound over all $i$ we obtain the success probability $1-|\T|n^{-12}\ge 1-n^{-11}$ as desired.

\paragraph{Part 2.} 
To show that $\mathcal{\hat{S}}\subseteq \mathcal{S}'$, assume by contradiction that there is $X\in\mathcal{\hat{S}}$ but $X\notin\mathcal{S}'$. By Line~\ref{line7} of \Cref{alg:compute good splits}, we know that $X$ is a $\T$-Steiner mincut. By Line~\ref{line13}, we know that $2\le |X\cap \T| \le \frac14|\T|$. 
Using \Cref{fact-for-case2}, we know that there must exist a part $T_i$ such that $X\cap \T\subseteq T_i$.
This implies that $X\subsetneq X'_i$.
However, from Part 1 we knew that with probability $1-n^{-12}$, $X'_i\in \mathcal{\hat{S}}$.
According to Line~\ref{line14}, $X$ will be removed from $\mathcal{S}$ so $X\notin \mathcal{\hat{S}}$, a contradiction.
\end{proof}

Recall that $\mathcal{S}'_{\text{pruned}}$ is defined by, first removing all $X'_i\in \mathcal{S}'$ where $|X'_i\cap \T|=1$, and then replace each $X'_i$ with $X_i=X'_i\setminus \cup_{j<i} X'_j$. Since these steps are identical to Line~\ref{line13} and Line~\ref{line15} of \Cref{alg:compute good splits}, we obtain the following corollary.

\begin{corollary}\label{lem:transition-part2-to-part3}
Suppose that \Cref{lem:steiner-mincut-case-2-part2} holds. Then, after the execution of Line~\ref{line15}, $\mathcal{S}=\mathcal{S}'_{\mathrm{pruned}}$.
\end{corollary}
\begin{proof}The post-processing that defines $S'_{\textrm{pruned}}$ is exactly the same as Line~\ref{line15}.
\end{proof}

Once we prove that with high probability, the returned collection $\mathcal{S}$ from \Cref{alg:compute good splits} is the same as $\mathcal{S}'_{\mathrm{pruned}}$, we can put all our effort now proving that $\mathcal{S}'_{\mathrm{pruned}}$ is exactly what we want for the divide-and-conquer algorithm. That is, we aim to show that $\mathcal{S}'_{\mathrm{pruned}}$ is a good split collection.

\begin{proposition}\label{prop-for-good-split-collection}
Each $X_i\in \mathcal{S}'_{\mathrm{pruned}}$ is a $\T$-split.    
\end{proposition}

\begin{proof}
First, $X_i$ is a $\T$-Steiner mincut by the definition $X_i=X'_i\setminus \cup_{j<i} X'_j$ and by posi-modularity (\Cref{lem:disjoint-posi-modularity}). Moreover, since each $X_i$ contains at least two terminal vertices so $X_i$ is a $\T$-split.
\end{proof}

\begin{lemma}\label{lem:steiner-mincut-case-2-part3}
Consider a particular pruned mincut $X_i\in\mathcal{S}'_{\mathrm{pruned}}$.
Then, the mincut on $H$ that separates $\phi(X_i\cap \T)$ with 
$\phi(\T\setminus X_i)$ is incident to $v$. That is, after the contraction of all $X_i$, the contracted graph has a cactus that is a star shape with $\ge 4$ leaves.
\end{lemma}

\begin{proof}
This is straightforward to check using the fact that $H$ is irredundant.
To bound the number of leaves,
we first deduce that $|\T|\ge 5$ --- this is because at least one split is found and there is no balanced split.
Moreover, using Assumption~\ref{assumption:2}, there can be at most one terminal being the center of the star.
Since each split contains strictly less than $\frac14|\T|$ terminals, we conclude that 
the star shape has at least $4$ leaves.
\end{proof}

\begin{lemma}\label{lemma:S-is-good-split-collection}
$\mathcal{S}'_{\mathrm{pruned}}$ is a good split collection.
\end{lemma}

\begin{proof}
It suffices to check that $\mathcal{S}'_{\text{pruned}}=\{X_i\}$ satisfies \Cref{def:good-splits}.
By \Cref{prop-for-good-split-collection}, we know that all $X_i$ are indeed $\T$-splits.
Furthermore, by the construction of $X_i$ we know that for any $i\neq j$, $X_i\cap X_j=\emptyset$. Thus,
condition (1) of \Cref{def:good-splits} is satisfied.

Now, since all $\T$-Steiner mincuts in $\mathcal{S}'_{\text{pruned}}$ are disjoint, the decomposition $\{(G_i, \T_i)\}$ induced by $\mathcal{S}$ is well-defined.
Without loss of generality we may apply simple refinements to $G$ in the order of $X_1, X_2, \ldots, X_{k}$. Let $(G_1, \T_1), \ldots, (G_{k+1}, \T_{k+1})$ be the decomposed graphs\footnote{Note that for $1\le i\le k$, we must have $\T_i\supsetneq T_i$: $\T_i$ contains exactly one more anchor vertex.}.
By \Cref{lem:steiner-mincut-case-2-part3}, we know that 
there exists a cactus representation that is a star shape for the very last graph $(G_{k+1}, \T_{k+1})$.
Therefore, by \Cref{def:good-decomposition}, $\{(G_i, \T_i)\}$ is a good decomposition.
This implies that condition (2) of \Cref{def:good-splits} holds and $\mathcal{S}$ is a good split collection.
\end{proof}

\begin{proof}[Proof of \Cref{lem:compute good splits}]
The correctness of
\Cref{lem:compute good splits}  follows from \Cref{lem:steiner-mincut-case-1} and \Cref{lemma:S-is-good-split-collection}.
To analyze the runtime, we first note that 
 \Cref{alg:compute good splits} involves $O(\log^2n)$ maximal isolating mincut computations, which contributes a total of $O((\log^3 n)\cdot \MaxFlow(3n, 4m))$ time by \Cref{thm:max-min-iso-cut}.
 Moreover, by \Cref{lem:total-size-max-min-iso-cut} we know that the total size of subsets after Line~\ref{line9} is at most $O(n\log^2 n)$.
 
It is not hard to see that each of the remaining steps can be implemented in time linear to the total size of the found mincuts, which is $O(n\log^2 n)$: in Line~\ref{line10} it suffices to scan through every subset in $\mathcal{S}$ and count the number of terminals; similar analysis holds for Line~\ref{line13}.
To implement Line~\ref{line14}, it suffices to scan through all the subsets $X_i\in \mathcal{S}$ and mark the size of $|X_i|$ at each terminal vertex $u\in \T\cap X_i$.
For each terminal vertex $u$, we simply select the largest sized $X_i$ that contains $u$ and get rid of all the others. This step also take $O(n\log^2 n)$ time.
 In Line~\ref{line15}, the algorithm initializes a boolean array $A$ and iterates through each set $X_1, X_2, \ldots, X_k$. For each set $X_i$, the algorithm checks for each vertex $u\in X_i$ whether or not $u$ has been set in $A$.
 If set, the algorithm discard $u$ from $X_i$. Otherwise, the algorithm marks $A[u]\gets\mathsf{true}$ and keeps $u$ in the set. Thus, it takes only a linear scan for Line~\ref{line15}. 
 In conclusion, the total runtime of \Cref{alg:compute good splits} is $O((\log^3 n)\cdot \MaxFlow(3n, 4m))$.
\end{proof}

\subsection{Returning a Correct Cactus}
\label{sec:compute-cactus-from-prime-decomposition-tree}

In this section, we prove \Cref{lem:merge-cactus}.
In particular, we will implement the procedures \textsc{TrivialCactus}($G, \T$), \textsc{StarCactus}($G, \T$), and \textsc{MergeCactus}($G, \T, \{H_i\}$).
Then, we will prove the correctness and analyze the runtime.
In this section, we assume that all the other divide and conquer steps (especially a good split collection is returned from Line~\ref{line:compute-good-collection} of \Cref{alg:divide-and-conquer-framework}) are correct.

\paragraph{Hollow 3-Stars vs 3-Cycle.}
A \emph{hollow 3-star} is an induced subgraph on $H$ that has 4 nodes where an empty node connects to exactly 3 other nodes.
Nagamochi and Kameda~\cite{nagamochi1994canonical} pointed out that if we replace this hollow 3-star induced subgraph with a 3-cycle (thereby removing the center empty node) from $H$, the resulting graph preserves the same set of mincuts on $H$.

To ensure a correct cactus being returned, we implement the procedures of~\Cref{alg:divide-and-conquer-framework} with the following invariant:

\begin{invariant}\label{inv:returned_cactus}
Let $(H, \phi)$ be the $\T$-Steiner cactus returned by any subproblem $(G, \T)$.
Then $H$ is irredundant, and does not contain a hollow 3-star as an induced subgraph.
\end{invariant}

Now, we start proving~\Cref{lem:merge-cactus}.

\paragraph{(Line~\ref{line:merge-cuts-1}) Trivial Cactus.} Since $|\T| \le 3$, by computing the mincut for every partition of $\T$ we obtain a $\T$-Steiner cactus in $O(\MaxFlow(n, m))$ time.
Notice that a returned cactus can be chosen to satisify~\Cref{inv:returned_cactus} as there are only $O(1)$ ways to construct such a cactus.

\paragraph{(Line~\ref{line:merge-cuts-2}) Star Cactus.} 
In this case, there is no $\T$-split that is a $\T$-Steiner mincut.
Since $|\T| \ge 4$ when Line~\ref{line:merge-cuts-2} is reached in \Cref{alg:divide-and-conquer-framework}, we claim that there exists a star shaped $\T$-Steiner cactus of $G$.

\begin{fact}
Let $G$ be a graph and $\T$ be a set of terminals with $|\T|\ge 4$.
Suppose that every $\T$-Steiner mincut is trivial (i.e., a $t$-isolating mincut of $\T$ for some $t\in \T$).
Then, there exists $(H, \phi)$, a $\T$-Steiner cactus of $G$ such that $H$ is a star graph.
\end{fact}

\begin{proof}
Let $A\subseteq \T$ be the subset of terminals $t$ whose $t$-isolating mincut is a $\T$-Steiner mincut. Construct $H$ as a star graph with $A$ being the set of leaves. The rest vertices in $\T\setminus A$ (possibly empty) are mapped to the center of $H$. Then $H$ preserves all $\T$-Steiner mincuts.
\end{proof}

With the assumption~\ref{assumption:2}, we know that there is at most one terminal vertex $t$ whose $t$-isolating mincut of $\T$ is not a $\T$-Steiner mincut.
Hence, after invoking one Isolating Cut Lemma~(\Cref{thm:isocut}) and checking the isolating mincut values, the \textsc{StarCactus} procedure is able to return a $\T$-Steiner cactus of $G$ in $O(\log|\T|\cdot\MaxFlow(2n, 2m))$ time.
By assumption~\ref{assumption:2} and $|\T|\ge 4$, no hollow 3-star can be formed and thus  \Cref{inv:returned_cactus} holds.

\paragraph{(Line~\ref{line:merge-cuts-3}, Case 1.) Merging from a Balanced Split.}
There are two cases when Line~\ref{line:merge-cuts-3} is reached, depending on whether a balanced split is found or not.
Suppose that a balanced split is found so the good split collection contains exactly one $\T$-split $|\mathcal{S}|=1$.

The \textsc{MergeCactus} procedure relies on the following useful observation.

\begin{fact}\label{lem:anchor-vertex-always-at-boundary}
Let $(G, \T)$ be a subproblem and let $t\in \T$ be an \emph{anchor vertex} generated in some ancestor problems.
Let $(H, \phi)$ be any irredundant $\T$-Steiner cactus of $G$.
Then, the node $\phi(t)$ on $H$ has either degree 1 (a leaf), or degree 2 but in a cycle.
\end{fact}

\begin{proof}
Since $t$ is an anchor vertex, $\{t\}$ is a $\T$-Steiner mincut of $G$.
Hence, there exists a mincut on $H$ that separates $\phi(t)$ with $\phi(\T\setminus \{t\})$. Since $H$ is irredundant, the $\phi(t)$-side of the mincut must be a single node (i.e., $\phi(t)$ itself), and the statement follows.
\end{proof}

Now, let us assume that the balanced split in a good split collection $\mathcal{S}$ decomposes the graph $G$ into two subproblems $(G_1, \T_1)$ and $(G_2, \T_2)$, with a shared anchor vertex $a\in \T_1\cap \T_2$.
Let $(H_1, \phi_1)$ and $(H_2, \phi_2)$ be the cactus returned from the two subproblems respectively.

Using \Cref{lem:anchor-vertex-always-at-boundary},
there are only constantly many situations to be handled:

\begin{description}
    \item[Leaf-Leaf Case.] If both anchor nodes $\phi_1(a)$ and $\phi_2(a)$ have degree 1, say edge $(x, \phi_1(a))$ in $H_1$ and edge $(\phi_2(a), y)$ in $H_2$, then the procedure forms the merged cactus by simply connecting $H_1$ and $H_2$ with an edge $(x, y)$ and then delete the anchor nodes $\phi_1(a)$ and $\phi_2(a)$. 
    Notice that \Cref{inv:returned_cactus} holds since $(x, y)$ cannot be further contracted, and this operation does not produce a hollow 3-star, simply because $|\T|\ge 4$.
    \item[Leaf-Cycle Case.] If the anchor vertex $a$ has degree 1 in one of the cactus and degree 2 in the other, without loss of generality, edge $(\phi_1(a), x)$ in $H_1$, edge $(\phi_2(a), y)$ and $(\phi_2(a), z)$ in $H_2$.
    Then the procedure simply connects $H_1$ and $H_2$ with edges $(x, y)$ and $(x, z)$ and then deletes the anchor nodes $\phi_1(a)$ and $\phi_2(a)$.
    \Cref{inv:returned_cactus} holds here too.
    \item[Cycle-Cycle Case.] The last case is nontrivial. Now both anchor nodes $\phi_1(a)$ and $\phi_2(a)$ have degree 2, say edges $(\phi_1(a), x_1)$ and $(\phi_1(a), y_1)$ in $H_1$ and edges $(\phi_2(a), x_2)$ and $(\phi_2(a), y_2)$ in $H_j$.
    We need to take care of the relation between these two cycles, as in this case there may be missing $\T$-Steiner mincuts.
    There will be three possible outcomes in the merged cactus, and the algorithm has to identify the correct formulation.
    \begin{enumerate}
        \item The two cycles are separated from each other in the cactus, and there is an empty node in the middle.
        \item They form one larger cycle together, concatenated by edges $(x_1,x_2)$ and $(y_1, y_2)$.
        \item They form one larger cycle together, concatenated by edges $(x_1,y_2)$ and $(y_1, x_2)$.
    \end{enumerate}
\end{description}

\paragraph{Test via Max-Flows.}
Fortunately, it is possible to distinguish the three cases mentioned above.
All we need to do is to obtain the mincut values between terminal sets $\{x_1,x_2\}$ and $\{y_1,y_2\}$, and between terminal sets $\{x_1,y_2\}$ and $\{x_2,y_1\}$ respectively.
This can be done by running $st$-mincut oracles on $G$ with the terminal sets $\{x_1,x_2\}$ (or $\{x_1,y_2\}$) contracted to $s$ and $\{y_1,y_2\}$ (or $\{x_2,y_1\}$ respectively) contracted to $t$.
If none of them equals to the known value $\lambda_G(\T)$, then it is the case (1).
Otherwise, $\{x_1,x_2\}$ or $\{x_1,y_2\}$ is a $\T$-Steiner mincut corresponding to case (2) and (3) respectively. 
This can be done in $O(\MaxFlow(n, m))$ time.

\paragraph{Correctness.} Here we briefly prove that the combined cactus is indeed a $\T$-Steiner cactus of $G$.
Let $H'$ be some irredundant $\T$-Steiner cactus of $G$. 
Let $\mathcal{S}=\{X\}$ and let $a\in\T_1\cap \T_2$ be the anchor vertex.
Since $H'$ is irredundant, there is a unique mincut on $H'$ that separates $\phi(X\cap \T)$ and $\phi(\T\setminus X)$.
It is straightforward to check that if $X$ corresponds to a 1-edge-cut of $H'$, then every mincut in $G$ is preserved in the subproblems already.
If $X$ corresponds to a $2$-edge-cut on $H'$, then at least one anchor node $\phi_1(a)$ or $\phi_2(a)$ must be in a cycle in the returned cactus. In this case, the above procedure recovers the cycle correctly.
\Cref{inv:returned_cactus} holds for the cycle-cycle case too as no new star can be formed.

\paragraph{(Line~\ref{line:merge-cuts-3}, Case 2.) Merging When There Is No Balanced Split.}
The other case when invoking Line~\ref{line:merge-cuts-3} is that no balanced split exists and thus a good split collection with one or more splits is returned.
Suppose that $|\mathcal{S}|=\ell \ge 1$ and the graph is decomposed into $\ell+1$ subproblems $\{(G_i, \T_i)\}$, with the last subgraph $(G_{\ell+1}, \T_{\ell+1})$ containing all the anchor vertices generated in this recursion step.
By \Cref{lem:steiner-mincut-case-2-part3} and the algorithm, the returned cactus from the subproblem $(G_{\ell+1}, \T_{\ell+1})$ must be a star graph of at least $4$ leaves. Let $H_{\mathrm{center}}$ be this star graph.

The algorithm attaches each cactus $H_i$ from other subproblems $(G_i, \T_i)$ to the star graph $H_{\mathrm{center}}$ via the {\textbf{Leaf-Leaf Case}} or the {\textbf{Leaf-Cycle Case}}.
Since no tests are required, this step can be done in linear time $O(|V(G)|+|E(G)|)$.
The correctness argument is the same as the balanced-split case since we can view this process as sequentially merging two cactus at a time.
For the same reason, \Cref{inv:returned_cactus} holds too.

\paragraph{Conclusion in Runtime.} We conclude that the bottleneck of the runtime happens whenever the algorithm invokes the \textsc{StarCactus} procedure, which invokes one isolating cut algorithm and runs in $O(\log|\T|\cdot \MaxFlow(2n, 2m))$ time.
All other operations require at most one max flow procedure so \textsc{MergeCactus} takes up to $O(\MaxFlow(n, m))$ time.

\section{Steiner Hypercactus Construction}
\label{sec:hypercactus}
\label{sec:hypercactus-representation-for-hyperedge-mincuts}

In this section, we generalize our algorithms in \Cref{sec:max isocut,sec:steiner-cactus} to hypergraphs; we give an almost-linear time construction of a \emph{Steiner hypercactus representation} that succinctly represents all Steiner hyperedge mincuts. 
This is the first almost-linear time algorithm even for normal hypercactus representation.

\paragraph{Preliminaries on (Steiner) Hypercactus Representation.}
A hypergraph $G = (V,E)$ consists of a vertex set $V$ and a hyperedge set $E$ where each edge $e$ is a subset of vertices. Let $n = |V|, m = |E|$ and $p = \sum_{e\in E}|e|$.
Given a subset of vertices $X\subseteq V$, the \emph{induced subhypergraph} $G[X]$ is defined by removing all outside vertices $V\setminus X$ and all the incident edges, i.e. $G[X] = (X, E_X)$ where $E_X = \{ e\in E \mid \forall v\in e, v\in X \}$.
Let $\T\subseteq V$ be a terminal vertex set. A $\T$-Steiner cut is a cut of $G$ that separates $\T$. A $\T$-Steiner mincut is a minimum valued $\T$-Steiner cut.

A \emph{hypercactus} is a hypergraph that satisfies the following properties:

\begin{itemize}[itemsep=0pt]
    \item $H$ is connected.
    \item Every rank-2 edge on $H$ is in at most one cycle, and this cycle must contain only rank-2 edges.
    \item For every hyperedge of rank $r>2$, removing this hyperedge partitions $H$ into \emph{exactly} $r$ connected components, where each connected component is also a hypercactus.
\end{itemize}

Fleiner and Jord{\'{a}}n \cite{FJ99} showed that there exists a hypercactus graph $H$ that represents all $\T$-Steiner hyperedge mincuts. 
\begin{definition}[Steiner Hypercactus, see also~\cite{FJ99}]\label{def:T-steiner-hypercactus}
Given a hypergraph $G$ and a terminal set $\T$,
a \emph{$\T$-Steiner hypercactus} $(H, \phi)$ is a weighted hypercactus graph $H$ with a mapping $\phi:\T\to V(H)$ such that 
(1) edges in a cycle have weights $\lambda_G(\T)/2$ and edges or hyperedges have weights $\lambda_G(\T)$, and (2)
an $A$-mincut of $\T$ is a $\T$-Steiner mincut if and only if a global mincut on $H$ separates $\phi(A)$ and $\phi(\T\setminus A)$.
\end{definition}

\paragraph{Our Result.}
Our main result is an algorithm for constructing Steiner hypercactus for any hypergraphs using polylogarithmic maxflow calls.

\begin{theorem}\label{thm:steiner-hypercactus-main-hypergraph}
Let $G$ be a hypergraph with $n$ vertices, $m$ edges and $p$ total volume of edges.
Let $\T$ be a set of terminals.
There exists a randomized Monte Carlo algorithm such that,
with probability $1-8n^{-10}$, the algorithm correctly computes a $\T$-Steiner hypercactus
in $O(\log^4 n)\cdot \mathrm{MaxFlow}(O(n+m), O(p + n\log |\T|))$ time.
\end{theorem}

The rest of the section is organized as follows. 
In \Cref{sec:max isocut hypergraph}, we show that by carefully modifying a definition of $A$-cuts in a subtle way, we can generalize our maximal isolating mincuts algorithm from \Cref{sec:max isocut} to hypergraphs.
Then, we present our divide-and-conquer algorithm for constructing Steiner cactus in \Cref{sec:divide and conquer hypergraph}, generalizing our algorithm in \Cref{sec:steiner-cactus}.

\subsection{Maximal Isolating Mincuts on Hypergraphs}
\label{sec:max isocut hypergraph}

\label{sec:max isocut hyper}
\label{sec:maximal-isolating-hyperedge-mincuts}

To overcome the first challenge to fast algorithms for computing maximal isolating mincuts in hypergraphs as discussed in \Cref{sec:techniques}, we carefully give a new definition of $A$-cuts of $\T$ in~\Cref{def:A-mincut hypergraph} with additional \emph{connectivity constraint}. Then, we give a generalization of \Cref{thm:max-min-iso-cut} in~\Cref{thm:max-min-iso-cut-hypergraph}.

\begin{definition}
\label{def:A-mincut hypergraph}
Let $G=(V, E)$ be a hypergraph and $\T\subseteq V$.
For any proper subset of terminals $A\subsetneq \T$, a cut $X$ is \emph{$A$-cut} of $T$ if it satisfies the following conditions:
\begin{itemize}[itemsep=0pt]
    \item The cut $(X, V\setminus X)$ separates $A$ and $\T\setminus A$, with $A\subseteq X$.
    \item After removing the boundary edges $\partial X$, for any $u\in X$, there exists a path from $u$ to some vertex $v\in A$.
    That is, $(G/A)[X]$ is connected.
\end{itemize}
\end{definition}
We use the same terminology for related concepts. An \emph{$A$-mincut} of $\T$ is a minimum valued $A$-cut of $\T$, denoted as $X_A$.  For any vertex $t\in \T$, a cut $X_t$ is \emph{$t$-isolating mincut} of $\T$ if $X_t$ is a $\{t\}$-mincut of $\T$.
We say that an $A$-mincut $X_A$ of $\T$ is \emph{maximal} (resp. minimal), if for any other $A$-mincut $X'_A$ of $\T$, we have $X_A\supseteq X'_A$ (resp. $X_A\subseteq X'_A$).

By crucially exploiting  \Cref{def:A-mincut hypergraph}, we are able to generalize the maximal isolating mincuts algorithm from \Cref{thm:max-min-iso-cut} to hypergraphs.

\begin{theorem}\label{thm:max-min-iso-cut-hypergraph}
There exists an algorithm that, given an undirected weighted hypergraph $G=(V, E)$ and a terminal set $\T\subseteq V$, in $O(\log|\T|)\cdot\mathrm{MaxFlow}(O(n+m),O(p))$ time computes the maximal $v$-isolating mincuts of $\T$ for all terminals $v\in \T$.
\end{theorem}

The new definition of $A$-cuts leads to some inconveniences. For example, an $A$-mincut of $\T$ may not be a $(\T\setminus A)$-mincut of $\T$ in a hypergraph. 
Also, suppose $X_A$ and $X_B$ are $A$-mincut and $B$-mincut, respectively, where $A \subseteq B$.
In normal graphs, $X_A \cap X_B$ would be an $A$-mincut by \Cref{lem:steiner-modularity}, but in hypergraphs $X_A \cap X_B$ might not be $A$-mincut with respect to our definition because $(G/A)[X_A \cap X_B]$ might not be connected. Anyhow, this inconvenience is easy to deal with.

The technical reason why we need to work with this new definition is that it allows us to show the hypergraph version of the Pairwise Intersection Only \Cref{lem:three-crossing-isolating-mincuts}, which is the key to efficiency.
Given this lemma, we verify that all ideas in \Cref{sec:max isocut} indeed generalize to hypergraphs and prove \Cref{thm:max-min-iso-cut-hypergraph} in \Cref{sec:max iso mincut proof}. 
We need to reprove everything since we are working with the new basic definition.

\subsection{Our Divide and Conquer Framework}
\label{sec:divide and conquer hypergraph}

In this section, we finally give an almost-linear time construction for Steiner hypercactus, based on the maximal isolating mincut algorithms for hypergraphs from \Cref{sec:maximal-isolating-hyperedge-mincuts}. 
Let us describe our algorithm. 

\paragraph{Preprocessing.}
First, we preprocess the hypergraphs in the same way that we did for  normal graphs. 
After preprocessing in $O(\log^2 n)\cdot\MaxFlow(O(n+m), O(p))$ time, we may assume that for any two terminal vertices there is a $\T$-Steiner mincut that separates them (see \Cref{lem:preprocessing-hypergraph}.)

\paragraph{Modified \Cref{alg:divide-and-conquer-framework}.} Our main divide-and-conquer algorithm for Steiner hypercactus is based on the same framework as \Cref{sec:divide and conquer framework}.
Recall that a hypergraph consisting of a single (weighted) hyperedge containing all vertices is called a \emph{brittle}.
We modify \Cref{alg:divide-and-conquer-framework} to make it compute a $\T$-Steiner hypercactus as follows. 
\begin{enumerate}
    \item In Line~\ref{line:compute-good-collection}, we called \Cref{alg:compute good splits} to return a good split collection. In Line~\ref{line:alg-good-split call max iso} of \Cref{alg:compute good splits}, we now use the maximal isolating mincut algorithm on hypergraphs (\Cref{thm:max-min-iso-cut-hypergraph}) instead. 
    \item  The hypercactus returned from the 
\textsc{StarCactus} procedure --- the procedure now may return a brittle instead of a star so we now call it \textsc{StarOrBrittleCactus}, which will be discussed in \Cref{sec:returning-a-correct-cactus-hyergraph}.
    \item  The implementation of \textsc{MergeCactus} is changed and will be specified in \Cref{sec:returning-a-correct-cactus-hyergraph}.
\end{enumerate}

\paragraph{Key Property: Never Split Higher-rank Hyperedges.}
As discussed in \Cref{sec:techniques}, the second challenge to efficiently compute a Steiner hypercactus representation is because a hypercactus contains hyperedges of rank more than two. 
Splitting these higher-rank hyperedges leads to slow run time. 
Our key observation is that the Modified \Cref{alg:divide-and-conquer-framework} never splits these hyperedges in a non-trivial way (proved at the end of the section).
    
\begin{lemma}
\label{cor:brittle real corollary} 
Consider \emph{any} $\T$-Steiner hypercactus $(H, \phi)$.
Let $e$ be an hyperedge on $H$ of rank $\ge 3$.
Then,
the Modified \Cref{alg:divide-and-conquer-framework} above never finds a split such that $e$ gets decomposed to  at least two smaller hyperedges with rank $\ge 3$ in at least two subproblems.
\end{lemma}

Therefore, the algorithm never splits higher-rank hyperedges in a hypercactus. 
This motivates us to define a \emph{good decomposition} for hypergraphs in almost the same way as defined for normal graphs (\Cref{def:good-decomposition}), except that we treat brittles as a base case similar to stars.

\begin{definition}[Good Decomposition for Hypergraphs]\label{def:good-decomposition-hypergraph}
    Given a hypergraph $G$ and a set of terminal vertices $\T$,
    a decomposition $\mathcal{G}=\{(G_i, \T_i)\}$ of $G$ is said to be \emph{good} with respect to $\T$ if $\mathcal{G}$ has the following property. Let $\T_i$ be the set of terminal vertices in $G_i$.
    For all $i$ except at most one special index $i^*$, $|\T_i| \le \frac34 |\T|+1$, and there exists a Steiner cactus representation of $\T_{i^*}$ in $G_{i^*}$ that is a star or a brittle.
\end{definition}

Now, the definition of a \emph{good split collection} on a hypergraph is the same as \Cref{def:good-splits} on normal graph.
Similar to \Cref{sec:steiner-cactus}, there are two main steps in the analysis of the algorithms. 

First, 
\Cref{lem:compute good splits hypergraph} summarizes the result that computes a good split collection on hypergraphs.

\begin{lemma}
\label{lem:compute good splits hypergraph}
    Given a hypergraph $G = (V, E)$ and a set of terminals $\T$, there exists a randomized Monte Carlo algorithm such that, with probability $1-n^{-11}$, the algorithm returns a good split collection $\mathcal{S}$ in $O(\log^3 n)\cdot \MaxFlow(O(n+m), O(p))$ time.
\end{lemma}

Second, \Cref{lem:merge-cactus hypergraph} summarizes the result for returning a correct hypercactus.

\begin{lemma}\label{lem:merge-cactus hypergraph}
Fix a subproblem $(G=(V, E), \T)$ in \Cref{alg:divide-and-conquer-framework}.
Assume all splits generated from the subproblems $\{(G_i, \T_i)\}$ derived from $(G, \T)$ are good, and each subproblem returns a correct $\T_i$-Steiner hypercactus of $G_i$.
Then, 
with probability $1-n^{-11}$,
the procedures \textsc{TrivialCactus}, \textsc{StarOrBrittleCactus}, and \textsc{MergeCactus} returns a $\T$-Steiner hypercactus of $G$  in $O(\log|\T|)\cdot \MaxFlow(O(n+m), O(p))$ time. The randomization comes only from the almost-linear time max-flow algorithm.
\end{lemma}

With~\Cref{lem:compute good splits hypergraph} and~\Cref{lem:merge-cactus hypergraph}, 
we can prove our main result~\Cref{thm:steiner-hypercactus-main-hypergraph}. Since the proof goes in very similar way to the analogous theorem for normal graphs (\Cref{thm:steiner-cactus-main}), we defer the proof to \Cref{sec:omit hypergraph}. Now, we proceed to prove  \Cref{lem:compute good splits hypergraph} and \Cref{lem:merge-cactus hypergraph} in respective subsections.

\subsection{Computing a Good Split Collection}

Now we aim to prove \Cref{lem:compute good splits hypergraph} using the same approach as in  \Cref{sec:compute good splits collection}.
Specifically, we show that \Cref{alg:compute good splits} also works for hypergraph.

Recall that $H$ denotes a hypercactus.
Motivated by \Cref{cor:brittle real corollary},
we observe that our algorithm never decomposes a rank $\ge 4$ hyperedge on $H$ into two smaller hyperedges of rank $\ge 3$.
This leads us to consider \emph{singular hyperedge cut} only.
Define a \emph{singular hyperedge cut} $(Q, V(H)\setminus Q)$ of $H$ if there exists hyperedge $e\in H$ such that $|Q\cap e| = 1$.

We call a mincut of $H$ \emph{accessible} if it is a singular hyperedge cut, or all edges across the boundary are normal edges (either one edge or two edges in a cycle).
Note that with~\Cref{cor:brittle real corollary}, all $\T$-Steiner mincuts returned from~\Cref{alg:divide-and-conquer-framework} correspond to accessible mincuts on $H$.
Define a \emph{balanced edge-cut} on hypercactus $H$ to be an accessible mincut of $H$ 
such that the number of terminals on both sides is between $\frac14|\T|$ and $\frac34|\T|$.
Again, to show correctness of~\Cref{alg:compute good splits}, our analysis depends on whether or not a balanced edge-cut exists on $H$. 
\thatchaphol{The notion of accessible mincut is never used in this section. Is this correct? Or it actually must be used somehow when you exploit \Cref{cor:brittle real corollary}}

\paragraph{Case 1: Balanced Cuts Exist.} In the first case where there is a balanced edge-cut, The proof is exactly the same as normal graph, except that we need to plug in a hypergraph version of \Cref{lem:sample-two-vertices}.

\begin{lemma}\label{lem:sample-two-vertices hypergraph}
Let $G$ be the hypergraph with a set $\T$ of terminals that satisfies Assumption~\ref{assumption:2}.
Let $r\in\T$ be a terminal such that any $r$-isolating mincut is a $\T$-mincut.
If we sample terminals $u, v\in \T-\{r\}$ uniformly at random,
then with probability at least $1/4$, 
maximal $\{r\}$-isolating mincut of $\T' = \{u,v,r\}$ has at most
$\frac12|\T|$ terminals.
\end{lemma}

\begin{proof}
Let $(H, \phi)$ be a $\T$-Steiner hypercactus of $G$. By Assumption~\ref{assumption:2} every vertex $v\in \T$ will be mapped to different vertices in $H$.
From the assumption that $r$-isolating mincut is a $\T$-mincut, we know that $\phi(r)$ is a \emph{leaf} node on $H$, i.e. either has degree 1 (may connected to a normal edge or a hyperedge) or has degree 2 within a cycle in $H$.
Consider a specialized DFS traversal of $H$ starting from $\phi(r)$: where upon visiting a vertex from an cycle edge, the DFS traversal always tends to choose any edge that leaves the cycle; upon first visiting a vertex from a specific hyperedge, the DFS traversal will visit the other veritces in the hyperedge with arbitrary order.
Let $(r, v_1, v_2, \ldots, v_{|\T|-1})$ be the a permutation of $\T$ where $(\phi(r), \phi(v_1), \phi(v_2), \ldots, \phi(v_{|\T|-1}))$ is the order (subsequence) of visited vertices by the DFS traversal, i.e. the pre-order.
Then for any two indices $i$ and $j$ such that $1\le i < j\le |\T|$, we will show that maximal $r$-isolating mincut of $\{ v_i, v_j, r \}$ must not contain any vertices in $\{v_i, v_{i+1}, \ldots, v_j\}$ by \Cref{cor:brittle real corollary}
and hence the result follows by counting the fraction of pairs (at least $1/4$) whose position in the permutation differs by at least $\frac12|\T|$.

It remains to show that maximal $r$-isolating mincut of $\{ v_i, v_j, r \}$ is disjoint with $\{v_i, v_{i+1}, \ldots, v_j\}$. There are three cases of the maximal $r$-isolating mincut cutting the hypercactus: (1) on a normal edge (2) on cycle edges (3) on a hyperedge. The former two cases are identical to the normal graph. For the third case, \Cref{cor:brittle real corollary} implies that the maximal $r$-isolating mincut $X_r$ contains either 1 or $|e|-1$ vertices of this hyperedge $e$ (not 0 or all the vertices since it cuts through this hyperedge). If $X_r$ contains one node $u$ in this hyperedge, then $u$ is the first one visited by the DFS traversal starting from root $r$. Therefore in this case, all the nodes in subtree rooted at $u$ are not contained in $X_r$ except $u$, while the subtree contains the entire set $\{v_i, v_{i+1}, \ldots, v_j\}$, and hence imples the claim. If $X_r$ contains all nodes in this hyperedge except one node $u$, then the subtree rooted at $u$ are not contained in $X_r$, and again implies that the set $\{v_i, v_{i+1}, \ldots, v_j\}$ is disjoint with $X_r$.
\end{proof}

\begin{lemma}\label{lem:steiner-mincut-case-1 hypergraph}
Let $G$ be the hypergraph with terminal set $\T$ and let $(H, \phi)$ be a $\T$-Steiner hypercactus of $G$.
Suppose there is a balanced edge-cut on $H$.
Then, with probability
$1-n^{-11}$
there is a balanced split in $\mathcal{S}$ returned from \Cref{alg:compute good splits}.
\end{lemma}

\begin{proof}[Proof Sketch.]
By plugging in \Cref{lem:sample-two-vertices hypergraph} into the proof of \Cref{lem:steiner-mincut-case-1}, we get the proof of \Cref{lem:steiner-mincut-case-1 hypergraph}. 
\end{proof}

\paragraph{Case 2: No Balanced Cuts.} In the second case there is no balanced edge-cut. Note that the definition of irredundant $\T$-Steiner hypercactus $H$ is the same as \Cref{def:irredundant-cactus}, since one can never contract a hyperedge in $H$ without losing Steiner mincuts.

\begin{lemma}\label{lem:steiner-mincut-case-2 hypergraph}
Let $G$ be a hypergraph with terminal set $\T$ and let $(H, \phi)$ be an irredundant $\T$-Steiner hypercactus of $G$.
Suppose there is no balanced edge-cut on $H$, then 
with probability $1-n^{-11}$ the following statements holds:
\begin{enumerate}[itemsep=0pt]
    \item Either (1)\label{lem5.11-part1-1} there exists a unique centroid node $v$ on $H$ whose all incident 1-edges and 2-edges from the same cycle corresponds to $\T$-Steiner mincuts of at most $\frac14|\T|$ terminals, or (2) there exists a unique hyperedge $e$ on $H$ such that each connected component in $H-e$ corresponds to $\T$-Steiner mincuts of at most $\frac14|\T|$ terminals.
    \item Define $\mathcal{S}'_{\mathrm{pruned}}$ in the same way as in \Cref{sec:compute good splits collection}. 
    Let $\mathcal{S}=\{X_i\}$ be the returned collection of $\T$-Steiner mincuts from \Cref{alg:compute good splits}.
    Then $\mathcal{S}=\mathcal{S}'_{\mathrm{pruned}}$.
    (Recall that $\mathcal{S}'_{\mathrm{pruned}}$ contains all $\T$-Steiner mincuts for each component $T_i$ described in part 1 as long as it contains at least $2$ terminals).
    \item Moreover, the mincut on $H$ that separates $\phi(X_i)$ with the rest vertices is incident to $v$. That is, after the contraction of all $X_i$, the contracted graph has a hypercactus that is a star shape with $\ge 4$ leaves or a brittle containing $\ge 4$ nodes. 
    \item $\mathcal{S}$ is a good split collection.
\end{enumerate}
\end{lemma}

\begin{proof}~
\paragraph{Part 1.} For each hyperedge $e\in H$, replace $e$ by a node $v$ and a bunch of normal edges $\{ (u,v)\mid u\in e \}$, and denote this cactus as $H'$. Then the uniqueness of centriod of $H'$ follows the same proof as \Cref{lem:steiner-mincut-case-2-part1}.
\item There is no balanced edge-cut on $H$, then by the definition of balanced edge-cut there must be a centroid node $v$ on $H'$ such that by removing $v$ the cactus $H'$ shattered into connected components.
Then centroid node $v$ mapped to a node or a hyperedge in $H$, corresponds to (1) and (2) respectively.
        
\paragraph{Part 2.} 
Let $\{T_i\}$ be the partition of all terminal vertices (possibly excluding $v$ if $v$ is non-empty) within each connected components.
Then by the assumption where no balanced cut exists, we have $\cup_i T_i = \T\setminus \{v\}, \forall i, |T_i|\le \frac14 |\T|$. 
Similar to~\Cref{lem:avoid-centroid} (but now the ``centroid'' $v$ could be either a node or a hyperedge), if three terminal vertices $w, x, y$ are sampled such that $\phi(w), \phi(x)$, and $\phi(y)$ belongs to three distinct connected component in $H-v$, then the maximal $w$-mincut of $\{w, x, y\}$ (denoted as $X_w$) has a corresponding mincut in $H$ that is exactly the connected component $W$ of $\phi(w)$ in $H-v$.

It is straightforward to check from definition that any $A$-mincut of $\T$ will be corresponding to an accessible mincut on $H$.\zhongtian{refer to \Cref{lem:unique mincut in brittle}}
By \Cref{cor:brittle real corollary}, our algorithm finds $X_w$, and it does not contain any vertex outside $\phi^{-1}(W)$.

Now, again
by \Cref{cor:brittle real corollary}, the algorithm always finds accessible mincuts.
Note that \Cref{fact-for-case2} and \Cref{prop-for-case-2} work for accessible mincuts.
Thus,
the algorithm ensures that $X_w$ stays maximal in Line~\ref{line14}.

The rest analysis about lower bounding the probability is the same as \Cref{lem:steiner-mincut-case-2-part2}.
Hence, with desired probability the returned collection $\mathcal{S}$ contains all desired $\T$-Steiner mincuts.

\paragraph{Part 3.}
By Part 2, $\mathcal{S}=\mathcal{S}'_{\mathrm{pruned}}$.
This is straightforward to check using the fact that $H$ is irredundant. The star case is proved in \Cref{lem:steiner-mincut-case-2-part3}.
To bound the number of nodes in the brittle, recall that each split contains strictly less than $\frac14|\T|$ terminals, so we conclude that 
the brittle has at least $5$ nodes.

\paragraph{Part 4.} By Part 2, $\mathcal{S}=\mathcal{S}'_{\mathrm{pruned}}$.
If \hyperref[lem5.11-part1-1]{case (1)} in Part 1 is true (i.e., there exists a unique centroid node on $H$), then the arguments follow exactly the same as in normal graph. Otherwise, there exists a unique hyperedge $e$ on $H$ such that each connected component in $H-e$ corresponds to $\T$-Steiner mincuts of at most $\frac14|\T|$ terminals. By Line~\ref{line:enfore-disjoint} in \Cref{alg:compute good splits}, condition (1) of \Cref{def:good-splits} is satisfied.

Now, since all $\T$-Steiner mincuts in $\mathcal{S}$ are disjoint\zhongtian{I will show ShangEn they are connected. I bet 1 dollar here.}, the decomposition $\{(G_i, \T_i)\}$ induced by $\mathcal{S}$ is well-defined.
Without loss of generality we may apply simple refinements to $G$ in the order of $X_1, X_2, \ldots, X_{k}$.
With this ordering, for each $1\le i\le k$, we have $|\T_i|\le |\T\cap X_i|+1 \le \frac14|\T|+1$.
Moreover, by Part 3, we know that there exists a cactus representation that is a brittle shape for the very last graph $(G_{k+1}, \T_{k+1})$.
Thus, by \Cref{def:good-decomposition-hypergraph}, $\{(G_i, \T_i)\}$ is a good decomposition.
This implies that condition (2) \Cref{def:good-splits} holds and $\mathcal{S}$ is a good split collection.
\end{proof}

\begin{proof}[Proof of \Cref{lem:compute good splits hypergraph}]
The correctness of
\Cref{lem:compute good splits hypergraph} follows from \Cref{lem:steiner-mincut-case-1 hypergraph} and \Cref{lem:steiner-mincut-case-2 hypergraph}.
It is straightforward to check (with the proof of~\Cref{lem:compute good splits}) that the runtime of Modified~\Cref{alg:compute good splits} is $O(\log^3 n)\cdot \MaxFlow(O(n+m), O(p))$. 
\end{proof}

\subsection{Returning a Correct Hypercactus}\label{sec:returning-a-correct-cactus-hyergraph}

In this subsection, we prove \Cref{lem:merge-cactus hypergraph}. In particular, we implement the procedures \textsc{TrivialCactus}($G, \T$), \textsc{MergeCactus}($G, \T, \{H_i\}$), and \textsc{StarOrBrittleCactus}($G, \T$). The first procedure is exactly the same as stated in \Cref{sec:compute-cactus-from-prime-decomposition-tree}. 

To ensure a correct hypercactus is returned, we implement the procedures of~\Cref{alg:divide-and-conquer-framework} with the following invariant:

\begin{invariant}\label{inv:returned-hypercactus}
Let $(H, \phi)$ be the $\T$-Steiner hypercactus returned by any subproblem $(G, \T)$.
Then $H$ is irredundant and does not contain a hollow 3-star as an induced subgraph.
\end{invariant}

\paragraph{Star or Brittle Cactus.} 
In this case, there is no $\T$-split that is an accessible $\T$-Steiner mincut.
Since $|\T| \ge 4$ when Line~\ref{line:merge-cuts-2} is reached in \Cref{alg:divide-and-conquer-framework}, we claim that there exists a star or brittle shaped $\T$-Steiner hypercactus of $G$.

\begin{fact}
Let $G$ be a hypergraph and $\T$ be a set of terminals with $|\T|\ge 4$.
Suppose that every accessible $\T$-Steiner mincut is trivial (i.e., a $t$-isolating mincut of $\T$ for some $t\in \T$).
Then, there exists $(H, \phi)$, a $\T$-Steiner hypercactus of $G$ such that $H$ is a star graph or a brittle.
\end{fact}

This fact follows from the definition of accessible mincut.
Let $A\subset \T$ be arbitrary subset with size $|A| = 2$. To distinguish the star case and brittle case, we simply test the mincut value between $A$ and $\T\setminus A$, which can be done in $O(\MaxFlow(n+m, p))$ time. And contructing the star cactus is the same as \Cref{sec:compute-cactus-from-prime-decomposition-tree}. The \textsc{StarOrBrittleCactus} procedure is able to return a $\T$-Steiner hypercactus of $G$ in $O(\log|\T|)\cdot\MaxFlow(O(n+m), O(p))$ time.
By assumption~\ref{assumption:2} and $|\T|\ge 4$, no hollow 3-star can be formed and thus  \Cref{inv:returned-hypercactus} holds.

\paragraph{(Line~\ref{line:merge-cuts-3}, Case 1.) Merge from a Balanced Split.}
There are two cases when Line~\ref{line:merge-cuts-3} is reached, depending on whether a balanced split is found or not.
Suppose that a balanced split is found so the good split collection contains exactly one $\T$-split $|\mathcal{S}|=1$.

The \textsc{MergeCactus} procedure rely on the following useful observation, which can be proved in the same way as \Cref{lem:anchor-vertex-always-at-boundary}.

\begin{fact}\label{lem:anchor-vertex-always-at-boundary hypergraph}
Let $(G, \T)$ be a subproblem and let $t\in \T$ be an \emph{anchor vertex} generated in some ancestor problems.
Let $(H, \phi)$ be any irredundant $\T$-Steiner hypercactus of $G$.
Then, the node $\phi(t)$ on $H$ has either degree 1, or degree 2 but in a cycle.
\end{fact}

Now, let us assume that the balanced split in a good split collection $\mathcal{S}$ decomposes the graph $G$ into two subproblems $(G_1, \T_1)$ and $(G_2, \T_2)$, with a shared anchor vertex $a\in \T_1\cap \T_2$.
Let $(H_1, \phi_1)$ and $(H_2, \phi_2)$ be the cactus returned from the two subproblems respectively. Using \Cref{lem:anchor-vertex-always-at-boundary},
there are only a constant situations to be handled, depending on whether $\phi_1(a)$ and $\phi_2(a)$ are leaves (connected to a normal edge), degree 1 nodes on brittle, or on cycle. The \textbf{Leaf-Leaf Case}, the \textbf{Leaf-Cycle Case}, and the \textbf{Cycle-Cycle Case} are implemented in \Cref{sec:compute-cactus-from-prime-decomposition-tree}. In addition, we only need to handle the following cases. 
Note that \Cref{inv:returned-hypercactus} still holds after these operations.

\paragraph{Brittle-Brittle Case.}
If both anchor nodes $\phi_1(a)$ and $\phi_2(a)$ have degree 1 and connected to hyperedges in both $G_1$ and $G_2$, then we simply make $a$ to be a normal node on hypercactus, and connect $G_1$ and $G_2$ via $a$.

\paragraph{Leaf-Brittle Case.}
If both anchor nodes $\phi_1(a)$ and $\phi_2(a)$ have degree 1, while one connected to a normal edge and the other connected to a hyperedge, say $(x, \phi_1(a))$ and hyperedge $e$ where $\phi_2(b)\in e$, then the procedure simply delete anchor vertex $a$ and replace hyperedge $e$ by hyperedge $(e\cup \{x\})\setminus \{a\}$, which connects $H_1$ and $H_2$.

\paragraph{Cycle-Brittle Case.}
If one anchor node connects to a cycle and the other connects to a hyperedge, say $(x, \phi_1(a)), (y, \phi_1(a))$ and hyperedge $e$ where $\phi_2(b)\in e$, then we simply make $a$ to be a normal node on hypercactus, and connect $G_1$ and $G_2$ via $a$.

\paragraph{(Line~\ref{line:merge-cuts-3}, Case 2.) Merging When There Is No Balanced Split.}

The other case when invoking Line~\ref{line:merge-cuts-3} is that no balanced split exists and thus a good split collection with one or more splits were returned.
Suppose that $|\mathcal{S}|=\ell \ge 1$ and the graph is decomposed into $\ell+1$ subproblems $\{(G_i, \T_i)\}$, with the last subgraph $(G_{\ell+1}, \T_{\ell+1})$ containing all the anchor vertices generated in this recursion step.
By part 3 of~\Cref{lem:steiner-mincut-case-2 hypergraph} and the algorithm, the returned cactus from the subproblem $(G_{\ell+1}, \T_{\ell+1})$ must be a star graph of at least $4$ leaves, or a brittle. Let $H_{\mathrm{center}}$ be this graph.

The algorithm attaches each cactus $H_i$ from other subproblems $(G_i, \T_i)$ to the center graph $H_{\mathrm{center}}$ via {\textbf{Leaf-Leaf Case}},  {\textbf{Leaf-Cycle Case}}, and {\textbf{Leaf-Brittle Case}} if $H_{\mathrm{center}}$ is a star graph; or via {\textbf{Leaf-Brittle Case}}, {\textbf{Cycle-Brittle Case}}, and {\textbf{Brittle-Brittle Case}} if $H_{\mathrm{center}}$ is a brittle.
Since no tests are required, this step can be done in linear time $O(|V(G)|+|E(G)|)$.
The correctness argument is the same as the balanced-split case, since we can view this process as sequentially merging two cactus at a time.
With the same reason, \Cref{inv:returned-hypercactus} holds too.

\subsection{Proof of \Cref{cor:brittle real corollary}}
\label{sec:structural properties of hyper mincuts}

In this section, we prove \Cref{cor:brittle real corollary}, the structural lemma that was crucial for our algorithm.
First, we shows that, although a hyperedge of rank $r$ on a hypercactus implies $2^r$ mincuts on $G$, almost all of these mincuts on $G$ have the same set of boundary edges.

\begin{lemma}
\label{lem:unique mincut in brittle}
Let $G=(V, E)$ be hypergraph and $\T$ be  a terminal vertex set with $|\T|\ge 4$.
Suppose that
there exists a brittle hypergraph $H$ that is a $\T$-Steiner hypercactus of $G$.
That is, for every proper subset $A\subsetneq \T$, 
there exists a $\T$-Steiner mincut that separates $A$ and $\T\setminus A$.
Then, there exists a unique set of hyperedges $E'\subseteq E$ such that, 
whenever $|A|\le |\T|-2$, $E'$ is
the set of boundary hyperedges to the maximal $A$-mincut $X_A$ of $\T$, i.e. $E'=\partial X_A$.
\end{lemma}

\begin{proof}
It suffices to prove that: $\partial X_a=\partial X_A$ for all $a\in A\subset \T$ such that $|A|\le |\T| - 2$, where $X_a$ (resp. $X_A$) is the maximal $a$-isolating (resp. $A$-) mincut of $\T$. We shall first prove a simpler statement $\partial X_a=\partial X_b$ for all $a, b\in \T$.

For any $a,b,c\in \T$, suppose $X_{ab}$, $X_{ac}$ and $X_{\overline{bc}}$ are the maximal $\{a,b\}$-mincut, maximal $\{a,c\}$-mincut and maximal $\T\setminus \{b,c\}$-mincut of $\T$ respectively. By assumption, $\C(X_{ab}) = \C(X_{ac}) = \C(X_{\overline{bc}}) = \lambda_G(T)$.

Let $Z_a = X_{ab}\cap X_{ac}\cap X_{\overline{bc}}$,
$Z_b = (X_{ab}\setminus X_{ac})\setminus X_{\overline{bc}}$,
$Z_c = (X_{ac}\setminus X_{ab})\setminus X_{\overline{bc}}$ and
$Z_{\overline{abc}} = (X_{\overline{bc}}\setminus X_{ab})\setminus X_{ac}$.
Then $\C(Z_a) = \C(Z_b) = \C(Z_c) = \C(Z_{\overline{abc}}) = \lambda_G(T)$ by submodularity and posi-modularity (\Cref{lem:submod}). Next we use the proof method similar to \Cref{lem:three-crossing-isolating-mincuts-on-hypergraph}. That is, the following invariant equality always holds.

\[
    \underbrace{(\C(X_{ab})+\C(X_{ac})+\C(X_{\overline{bc}})) - (\C(Z_b)+\C(Z_c)+\C(Z_{\overline{abc}}))}_{\texttt{(LHS)}} = 0 ~.
\]
Using the same argument with the proof of \Cref{lem:three-crossing-isolating-mincuts-on-hypergraph}, we have
\begin{enumerate}
    \item For any hyperedge $e$, the total contribution of the weight $e$ to the LHS of the equality is non-negative.
    \item\label{property2} For any hyperedge $e$ connecting $X_a$ such that $e$ contributes 0 to the LHS of the equality, either $e$ contributes to none of $\C(X_{ab})$, $\C(X_{ac})$ and $\C(X_{\overline{bc}})$ , or $e$ contributes to all of $\C(Z_b)$, $\C(Z_c)$, and $\C(Z_{\overline{abc}})$.
\end{enumerate}

Let $X_a$ be the maximal $a$-isolating mincut of $\T$. First observe that $X_a\subseteq Z_a$, otherwise contradicts to maximality of either $X_{ab}$, $X_{ac}$ or $X_{\overline{bc}}$ by Nesting \& Submodularity \Cref{lem:steiner-modularity-hypergraph}. So $X_a$ must be the connected component in $G[Z_a]$ connected to $a$, by the definition of maximal $\{a\}$-mincut of $\T$. Therefore $\partial X_a = \partial Z_a$.

For any edge $e$ in $\partial Z_a$, it always contributes to either $\C(X_{ab})$, $\C(X_{ac})$ or $\C(X_{\overline{bc}})$. By property~\ref{property2}, it will also contributes to all of $\C(Z_b)$, $\C(Z_c)$, and $\C(Z_{\overline{abc}})$, furthermore by the equality, contributes to all of $\C(X_{ab})$, $\C(X_{ac})$ and $\C(X_{\overline{bc}})$. Combining with the fact that $\C(X_{ab}) = \C(Z_a) = \lambda_G(T)$, we have $\partial X_{ab} = \partial Z_a = \partial X_a$.

Since $a, b$ are arbitrary terminals in $\T$, we also have $\partial X_{ab} = \partial X_b$ by swapping $a$ and $b$ in the argument above. Therefore, for any $a, b\in \T$, $\partial X_a = \partial X_b$ and the statement follows.

After replacing $b$ by $A\setminus \{a\}$ and using the same argument, we have $\partial X_a=\partial X_A$ for all $a\in A\subset \T$ such that $|A|\le |\T| - 2$.
\end{proof}

\begin{lemma}\label{def:cor6.4}\label{lem:never-split-brittle}
Consider a $\T$-Steiner hypercactus $(H, \phi)$ of $G$.
Let $\T'\subseteq\T$ be a subset of $\T$ with $|\T'|\ge 2$.
Let $e$ be a hyperedge on $H$ with rank $\ge 3$.
Then for all $t\in \T'$ such that the maximal $t$-isolating mincut $X_t$ of $\T'$ is a $\T'$-Steiner mincut,
any corresponding mincut that separates $\phi(X_t\cap \T)$ and $\phi(\T\setminus X_t)$ on $H$ includes $0$, $1$, $|e|-1$, or $|e|$ nodes in $e$.
\end{lemma}

\begin{proof}
For any $u\in e$, let $r(u) = \phi^{-1}(v)$ be a terminal vertex in $G$ where $v$ is some arbitrary fixed node in the connected component of $H$ connected to $u$ when removing $e$. There exists such $v$ with non-empty pre-image by the definition of hypercactus. Let $\hat T_0 = \{ r(u)\mid u\in e \}$, the hypercactus of $\hat T_0$ is a brittle.

Suppose a contradiction there exists $X_t$ to be a maximal $t$-isolating mincut of $\T'$, and mincut $(Q, V(H)\setminus Q)$ on $H$ that separates $\phi(X_t\cap \T)$ and $\phi(\T\setminus X_t)$,
such that $2\le |Q\cap e|\le |e|-2$.
Let $v\in e$ be the node in the same connected component with $\phi(t)$ when removing $e$ from $H$.
Define $\hat \T = (\hat T_0 \setminus \{r(v)\}) \cup \{t\}$.
Then, the hypercactus representation of $\hat \T$ is also a brittle.
Let $A = \{ r(u)\mid u\in (Q\cap e)\setminus \{v\}\}\cup \{t\}$ which is a subset of $\hat \T$.
Observe that the hyperedge $e$ on $H$ itself must be the \emph{only} mincut
that separates $\phi(X_A\cap \T)$ and $\phi(\T\setminus X_A)$ (resp. separates $Q$ and $V(H)\setminus Q$).
So, for consider any two vertices $a, b$ such that their mapped vertices on $H$ are in the same connected component of $H-e$, we have $a\in X_A$ if and only if $b\in X_A$.
Respectively, $a\in X_t$ if and only if $b\in X_t$.

Now we have $A\subseteq X_t$ since $r(u)\in A$ implies that $u\in (Q\cap e)$ so there is some $a\in X_t$ such that $\phi(a)$ is in the same component of $H-e$ with $u$, further implies that $r(u)\in X_t$.
Now we claim that $X_t=X_A$.
First, $X_t$ does not contain other terminals in $\hat \T\setminus A$, so $X_t\subseteq X_A$ (otherwise, $X_t\cup X_A$ is a larger sized (connected) $A$-mincut of $\T$ by submodularity).
On the other hand, since $A\subseteq X_t$ and that $X_A$ does not contain any terminal in $(\T'\setminus \{t\})$ (otherwise, $X_A$ contains a terminal $a\in \T'\setminus \{t\}$, let $u\in e$ be the node in the same connected component with $\phi(a)$ in $H-e$. Then, $r(u)\in A\subseteq X_t$, which further implies that $a\in X_t$, contradicting to the fact that $X_t$ separates $a$ and $t$.)
By the same submodularity argument we have $X_A\subseteq X_t$.

There exists $w$ be a vertex in $A$ other than $t$, since $|A|\ge 2$.
By \Cref{lem:unique mincut in brittle} ($|A|\le |e| - 2 = |\hat \T| - 2$), $\partial X_A = \partial X_w$ where $X_w$ is the maximal $w$-isolating mincut of $\hat T$, so $\partial X_t = \partial X_A = \partial X_w$. $\partial X_w$ separates $w$ and $t$ by definition, and $\partial X_t = \partial X_w$, contradicts to $t$ and $w$ are connected in $G[X_t]$.
\end{proof}

 \Cref{lem:never-split-brittle} directly implies \Cref{cor:brittle real corollary}, since the splits used in the Modified \Cref{alg:divide-and-conquer-framework} are maximal isolating mincuts of some terminal sets $\T'\subseteq \T$.
 \thatchaphol{explain 
 a bit more.}\zhongtian{done.}

\section{Conclusion and Open Problems}
We develop a new approach based on \emph{maximal isolating cuts} for computing the cactus representation of all mincuts that gives the first almost-linear time algorithms for computing Steiner hypercactus, which generalizes both Steiner cactus and hypercactus, each of which generalizes the standard cactus for global edge mincuts. 

A natural question is whether our framework works with even more generalized settings than hypergraph connectivity, such as \emph{element connectivity}. 
Let $U \subseteq V$ be a set of terminals. An element mincut between $s$ and $t$ is the smallest mixed cut $C \subset (E \cup (V-U))$ whose removal disconnects $s$ and $t$. 
There exists a hypercactus representation that captures all global element mincuts as well  \cite{FJ99} (in the same sense as Steiner cactus captures all Steiner mincuts). Given our result, one can also hope that it admits an almost-linear time construction.

However, element cuts do not fall into the setting of symmetric submodular set functions. Instead, they are captured by a more general notion of \emph{bisubmodular} set functions (see, e.g., \cite{CQ21}). To make our approach works, it seems we need to generalize the notion of posi-modularity for bisubmodular set functions, but it is unclear how to come up with the right definition. More concretely, what is the usable version of the Pairwise Insertion Only Lemma (\Cref{lem:three-crossing-isolating-mincuts}) for element cuts?

Since cactus representation exists for arbitrary symmetric submodular set functions \cite{FJ99}, it is also interesting whether there are algorithms with small \emph{query complexity} for cactus construction. Our algorithm carefully decomposes graphs into small pieces and works on each of them separately so that the total running time is almost-linear. This approach that works on small pieces in parallel does not seem to work with the setting for arbitrary symmetric submodular set functions where we count the number of queries.

\section*{Acknowledgement}
We thank the anonymous reviewers for their constructive suggestions.

\bibliographystyle{alpha}
\bibliography{main}

\appendix 
\input{appendix.tex}

\input{isocut_hypergraphs}

\input{omit_hypergraph}

\input{app}
\end{document}

%% file: header.tex
\usepackage{fullpage}
\usepackage[utf8]{inputenc}
\usepackage{amsmath,amsfonts,amsthm,amssymb}
\usepackage{setspace}
\usepackage{fancyhdr}
\usepackage{lastpage}
\usepackage{extramarks}
\usepackage{chngpage}
\usepackage{soul,color}
\usepackage{graphicx,float,wrapfig}
\usepackage{listings}
\usepackage{xcolor}
\usepackage{enumerate}
\usepackage{enumitem}
\usepackage{xspace}
\usepackage[T1]{fontenc}

\usepackage{xcolor}
\usepackage{nameref}

\definecolor{ForestGreen}{rgb}{0.1333,0.5451,0.1333}
\definecolor{DarkRed}{rgb}{0.65,0,0}
\definecolor{Red}{rgb}{1,0,0}
\usepackage[linktocpage=true,
pagebackref=true,colorlinks,
linkcolor=DarkRed,citecolor=ForestGreen,
bookmarks,bookmarksopen,bookmarksnumbered]
{hyperref}
\usepackage{cleveref}
\usepackage{algpseudocode}
\usepackage{algorithm}  
\usepackage{algorithmicx} 
\usepackage{verbatim}
\usepackage{todonotes}
\usepackage{appendix}

\makeatletter
\def\namedlabel#1#2{\begingroup
    #2%
    \def\@currentlabel{#2}%
    \phantomsection\label{#1}\endgroup
}
\makeatother

\usepackage{thmtools}

\declaretheorem[numberwithin=section]{theorem}
\declaretheorem[numberlike=theorem]{lemma}

\declaretheorem[numberlike=theorem]{fact}
\declaretheorem[numberlike=theorem]{proposition}

\declaretheorem[numberlike=theorem]{corollary}
\declaretheorem[numberlike=theorem,name=Corollary]{cor}

\declaretheorem[numberlike=theorem]{invariant} 
 
\declaretheorem[numberlike=theorem,style=definition]{definition}

\usepackage{mathrsfs}  

\newcommand{\C}{\mathcal{C}}
\newcommand{\T}{\mathcal{T}}
\newcommand{\G}{\mathcal{G}}
\newcommand{\MaxFlow}{\mathrm{MaxFlow}}
\newcommand{\poly}{\mathrm{poly}}
\global\long\def\Otil{\tilde{O}}
\global\long\def\Ohat{\widehat{O}}
\global\long\def\mincut{\mathrm{mincut}}

\Crefname{ALC@unique}{Line}{Lines}
\Crefname{fact}{Fact}{Facts}
\Crefname{invariant}{Invariant}{Invariants}

\ifdefined\ShowComment

\def\thatchaphol#1{\marginpar{$\leftarrow$\fbox{T}}\footnote{$\Rightarrow$~{\sf\textcolor{purple}{#1 --Thatchaphol}}}}

\def\zhongtian#1{\marginpar{$\leftarrow$\fbox{Z}}\footnote{$\Rightarrow$~{\sf\textcolor{blue}{#1 --Zhongtian}}}}

\else
\def\thatchaphol#1{}
\def\zhongtian#1{}

\fi

%% file: 0-abstract.tex
\begin{abstract}
A \emph{cactus} representation of a graph, introduced by Dinitz \emph{et al.}~in 1976, is an edge sparsifier of $O(n)$ size that exactly captures \emph{all} global minimum cuts of the graph. It is a central combinatorial object that has been a key ingredient in almost all algorithms for the connectivity augmentation problems and for maintaining minimum cuts under edge insertions (e.g.~{[}Naor \emph{et al.}~SICOMP'97{]}, {[}Cen \emph{et al.} SODA'22{]}, {[}Henzinger ICALP'95{]}). This sparsifier was generalized to \emph{Steiner cactus} for a vertex set $T$, which can be seen as a vertex sparsifier of $O(|T|)$ size that captures all partitions of $T$ corresponding to a $T$-Steiner minimum cut, and also \emph{hypercactus}, an analogous concept in hypergraphs. These generalizations further extend the applications of cactus to the Steiner and hypergraph settings. 

In a long line of work on fast constructions of cactus and its generalizations, a near-linear time construction of cactus was shown by Karger and Panigrahi {[}SODA'09{]}. Unfortunately, their technique based on tree packing inherently does not generalize. The state-of-the-art algorithms for Steiner cactus and hypercactus are still slower than linear time by a factor of $\Omega(|T|)$ {[}Dinitz and Vainshtein STOC'94{]} and $\Omega(n)$ {[}Chekuri and Xu SODA'17{]}, respectively. 

We show how to construct both Steiner cactus and hypercactus using polylogarithmic calls to max flow, which gives the first almost-linear time algorithms of both problems. The constructions immediately imply almost-linear-time connectivity augmentation algorithms in the Steiner and hypergraph settings, as well as speed up the incremental algorithm for maintaining minimum cuts in hypergraphs by a factor of $n$.

The key technique behind our result is a novel variant of the influential \emph{isolating mincut technique} {[}Li and Panigrahi FOCS'20, Abboud \emph{et al.~}STOC'21{]} which we called \emph{maximal isolating mincuts.} This technique makes the isolating mincuts to be ``more balanced'' which, we believe, will likely be useful in future applications. 
\end{abstract}

%% file: 1-intro.tex
\section{Introduction}

In a weighted undirected $G=(V,E)$ with $n$ vertices and $m$ edges, a \emph{global minimum cut} (or \emph{mincut}, for short) is a cut with minimum weight separating some pair of vertices. More than 40 years ago, Dinitz \emph{et al.}~\cite{dinits1976structure} showed that, even though $G$ may have as many as $\binom{n}{2}$ distinct mincuts, there exists a $O(n)$-size data structure called a \emph{cactus representation} or simply a \emph{cactus} of $G$ that captures \emph{all} mincuts of $G$ in a strong way as follows. 

A cactus of $G$ is a tuple $(H,\phi)$ where $H$ is a graph with $O(n)$ edges and $\phi:V(G)\rightarrow V(H)$ is a mapping such that, for any vertex set $A\subset V$, a mincut in $G$ separates $A$ and $V\setminus A$ iff a mincut in $H$ separates $\phi(A)$ and $\phi(V\setminus A)$. Hence, $H$ preserves all mincuts of $G$. Furthermore, $H$ and its mincuts are highly-structured: $H$ is a \emph{cactus graph}\footnote{A \emph{cactus graph} is a graph where each edge appears in at most one cycle.} with edge weights from $\{1,2\}$ and its mincut contains either two edges of weight 1 from the same cycle, or an edge of weight 2. Precisely by the ability of cactus to capture all mincuts via extremely simple structure, cactus has become the key ingredient in almost all algorithms for the connectivity augmentation problems \cite{gabow1991applications,naor1997fast,benczur2000augmenting,cen2022augmenting} and for maintaining mincuts on graphs undergoing edge insertions \cite{henzinger1997static,dinitz1998maintaining,goranci2018incremental}. Cactus can also be viewed as one of the first graph sparsifiers, predating other notions such as spanners \cite{althofer1993sparse}, cut sparsifiers \cite{benczur1996approximating}, and spectral sparsifiers \cite{spielman2011spectral}. 

\paragraph{Steiner cactus and hypercactus.}

To capture Steiner mincuts which are more general than global mincuts, a generalization of cactus called \emph{Steiner cactus} was introduced by Dinitz and Vainshtein \cite{DV94}.\footnote{A Steiner cactus was introduced as a core structure inside a more involved structure called \emph{carcass} \cite{DV94,dinitz19952,dinitz2000general}. The detailed proofs of the existence of Steiner cactus were given in \cite{dinitz1996cactus,fleischer1999building,Fleiner2009AQP}} Recall that, for any vertex set $\T\subseteq V$, a \emph{$\T$-Steiner mincut} is a cut with minimum weight separating some pair of vertices in $\T$. A \emph{$\T$-Steiner cactus} of $G$ is a tuple $(H,\phi)$ where $H$ is a graph with $O(|\T|)$ edges and $\phi:\T\rightarrow V(H)$ is a mapping such that, for any $A\subset\T$, a $\T$-Steiner mincut in $G$ separates $A$ and $\T\setminus A$ iff a mincut in $H$ separates $\phi(A)$ and $\phi(\T\setminus A)$. The graph $H$ is also a cactus graph with the same simple structure as in a normal cactus. Note that, when $\T=V$, $\T$-Steiner cactus of $G$ is simply a cactus of $G$. 

The notion of Steiner cactus can be placed nicely into a more modern concept of \emph{vertex sparsifiers} (also called \emph{mimicking networks}) \cite{hagerup1998characterizing,moitra2009approximation,leighton2010extensions,krauthgamer2013mimicking,khan2014mimicking}. The goal in this area is, given a terminal set $\T\subseteq V$, to construct a small graph $H$ and a mapping $\phi$ such that, for every $A\subset\T$, $\mincut_{G}(A,\T\setminus A)=\mincut_{H}(\phi(A),\phi(\T\setminus A))$ where $\mincut_{G}(X,Y)$ denotes the size of minimum cuts separating the sets $X$ and $Y$ in $G$. Unfortunately, there is a lower bound of $|E(H)|=\Omega(2^{|\T|})$ \cite{krauthgamer2013mimicking}, and perhaps the most prominent open problem in this area is, when $(1+\epsilon)$-approximation is allowed, whether the bound $|E(H)|=\poly(|\T|)$ is possible. Interestingly, Steiner cactus implies that the linear bound $|E(H)|=O(|\T|)$ is actually possible without any approximation when we restrict ourselves to the sets $A\subset\T$ separated by some $\T$-Steiner mincuts.

Another generalization of cactus is called \emph{hypercactus}.  
Cheng \cite{Cheng99} (and later \cite{FJ99}) showed an existence of hypercactus $(H,\phi)$ where $H$ is a hypergraph of linear size $\sum_{e\in E(H)}|e|=O(n)$, yet, for every set $A\subset V$, a mincut in $G$ separates $A$ and $V\setminus A$ iff a mincut in $H$ separates $\phi(A)$ and $\phi(V\setminus A)$. Similar to cactus, $H$ is highly structured: each mincut in $H$ contains either two size-2 edges of weight 1 from the same cycle of size-2 edges, or a (hyper)edge of weight 2. 
This compact data structure is perhaps even more surprising than cactus because the total size of a hypergraph can be exponential and a hypergraph may contain exponentially many distinct mincuts.\footnote{Consider a hypergraph $G=(V,E)$ with a single hyperedge containing all vertices, every cut $(S,V\setminus S)$ is a mincut. We note that if we consider set of hyperedges across a mincut, there are at most $O(n^2)$ distinct sets.}

Both Steiner cactus and hypercactus naturally extend the reach of applications of cactus. Cole~et~al.~\cite{cole2003fast} used a Steiner cactus to speed up the algorithm for the \emph{uniform survivable network} problem. They exploited the fact that Steiner cactus can be efficiently maintained under edge insertions between terminals. Hypercactus were used for hypergraph connectivity augmentation algorithms \cite{Cheng99}, and incremental algorithms for hypergraph mincuts \cite{gupta2019incremental}.

\paragraph{Fast Algorithms.}

Because of the elegance and utility of cactus and its generalization, a long line of work has been devoted on fast algorithms for constructing them. Historically, cactus construction has been much more challenging than a more well-known problem of computing a single mincut (e.g.~\cite{karger1993global,karger2000minimum,klimmek1996simple,mak2000fast}), as we need to capture the structure of \emph{all} mincuts. 

Karzanov and Timofeev \cite{karzanov1986efficient} outlined the first algorithm that constructs a cactus of a graph in $\Theta(n^{3})$ time. Their algorithm was parallelized by Naor and Vazirani \cite{naor1991representing} and refined by 
Nagamochi and Kameda \cite{nagamochi1994canonical}. Later on, faster algorithms were developed \cite{gabow1991applications,karger1996new,nagamochi2000fast,fleischer1999building} where the latter two algorithms ran in $\Otil(nm)$ time. Finally, the line of work culminated in a near-linear $\Otil(m)$-time algorithm by Karger and Panigrahi \cite{karger2009near}.\footnote{Throughout the paper, $\Otil(\cdot)$ hides $\textrm{polylog}(n)$ terms and $\Ohat(\cdot)$ hides additional $n^{o(1)}$ terms.}

The state-of-the-art constructions for  Steiner cactus and hypercactus are significantly slower. Dinitz and Vainshtein \cite{DV94,dinitz2000general} showed how to compute a $\T$-Steiner cactus using $\Theta(|\T|)$ max flow calls, which is $\Omega(m|\T|)$ time. Cole~et~al.~\cite{cole2003fast} showed an $\tilde{O}(m+\lambda_{G}(T)n)$-time algorithm on unweighted graphs where $\lambda_{G}(T)$ denotes the value of $\T$-Steiner mincut, but in general $\lambda_{G}(T)$ may be big especially in weighted graphs. To construct a hypercactus of a hypergraph with total size $p=\sum_{e\in E}|e|$, the only algorithm with explicit running time was by Chekuri and Xu \cite{CX17}, which takes $O(pn+n^{2}\log n)$ time. 

To summarize, the fastest constructions for Steiner cactus and hypercactus are still slower than linear time by a factor of $\Omega(|T|)$ and $\Omega(n)$, respectively. This suggests a natural question of how fast one can compute them. 

\subsection{Our Results}

We give a novel approach for constructing a cactus that generalizes to both Steiner cactus and hypercactus using polylogarithmic calls to max-flows. Let $\MaxFlow(m)$ denote the running time of solving max flow in a graph with $m$ edges. Since $\MaxFlow(m)=m^{1+o(1)}$ by \cite{chen2022maximum}, we obtain the first almost-linear time algorithms of both problems. The Steiner cactus algorithm is summarized below.
\begin{theorem}
\label{thm:cactus main}There is a randomized Monte-Carlo algorithm that, given an undirected weighted graph $G$ with $m$ edges and a terminal set $\T\subseteq V$, 
the algorithm computes a $\T$-Steiner cactus of $G$ in $\tilde{O}(\MaxFlow(O(m)))$ time w.h.p.
\end{theorem}

On hypergraphs, we even obtain an algorithm for computing a \emph{Steiner hypercactus}, which naturally generalizes both a hypercactus and a Steiner cactus (see \Cref{def:steiner-cactus} for the formal definition). The result is summarized below. 
\begin{theorem}
\label{thm:hypercactus main}There is a randomized Monte-Carlo algorithm that, given a weighted hypergraph $G$ with total size $p=\sum_{e\in E(G)}|e|$ and a terminal set $\T\subseteq V$, compute a $\T$-Steiner hypercactus of $G$ in $\Otil(\MaxFlow(O(p)))$ time w.h.p.
\end{theorem}

It is implicit in \cite{FJ99,CX17} that a Steiner hypercactus admits polynomial time algorithms, but their construction requires at least $n$ calls to max flows, which takes at least $\Omega(pn)$ time.
Even for the more special problem of computing hypercactus, the best-known construction by \cite{CX17} takes $\Otil(pn)$ time.\footnote{Their result is based on the \emph{MA-ordering} technique, which does not work well with Steiner mincuts.} \Cref{thm:hypercactus main} gives the first almost-linear time construction. 

As discussed above, a cactus and its generalizations are central objects in many algorithms. Consequently, our almost-linear constructions immediately imply several applications. We defer the definition of the problems and the proofs of these applications to \Cref{sec:applications}. 
\begin{cor}\label{cor:edge-augmentation}
There are randomized almost-linear time algorithms that can w.h.p. compute
\begin{itemize}[noitemsep,nolistsep]
\item the optimal solution of the Steiner connectivity augmentation problem\footnote{
Very recently, Cen \emph{et al.}~\cite{cen2023augmenting} independently showed an almost-linear time algorithm for the Steiner connectivity augmentation problem.}, and 
\item the optimal value of the hypergraph +1-Steiner-connectivity augmentation problem. 
\end{itemize}
\end{cor}
\Cref{cor:edge-augmentation} improved a polynomial algorithm for the hypergraph +1-connectivity by Cheng \cite{Cheng99} by speeding
up it to almost-linear time and generalizing it to the Steiner version. 

Lastly, we also improved the update time of the incremental algorithm for maintaining hypergraph mincuts from $O(\lambda n)$ \cite{gupta2019incremental} to of $\Ohat(\lambda)$.
\begin{cor}\label{cor:incremental-hypergraph-mincut}
There is an algorithm that, given an unweighted hypergraph $G=(V,E)$ undergoing hyperedge insertions, maintains a mincut in time $\Ohat(\lambda)$ amortized update time where $\lambda$ denotes the mincut value at the end of the updates. 
\end{cor}

\subsection{Techniques}
\label{sec:techniques}

Below, we explain why the two most promising techniques in the literature fail to solve our problems, which motivates our new algorithmic tool called \emph{maximal isolating mincuts}.

\paragraph{First Technical Barrier: Tree-packing Fails. }

Perhaps the most natural approach for devising fast algorithms for both Steiner cactus and hypercactus is to extend the techniques of the only known near-linear time algorithm by Karger and Panigrahi \cite{karger2009near} for a normal cactus. However, their algorithm relied on the \emph{tree packing technique} \cite{gabow1991matroid,karger2000minimum}, which requires solving the so-called \emph{2-respecting mincuts} problem. Although they devised an ingenious way to deal with 2-respecting mincuts, at least quadratic time is likely required for $k$-respecting mincuts for any $k>2$.\footnote{Abboud \emph{et al.}~\cite{AKLPST21} devised a new technique for dealing with $k$-respecting mincuts using $\log^{O(k)}n$ calls to max flow. The algorithm, however, is based on the isolating mincuts technique, which also fails to solve our problems. Furthermore, the exponential dependency on $k$ makes their technique futile for hypergraphs where $k=O(\log n)$.}

Now, in the Steiner and hypergraph settings, it turns out that one needs to compute $4$-respecting and $O(\log n)$-respecting mincuts, respectively, because fast algorithms for tree packing in these settings have worse quality by at least a factor of 2 \cite{mehlhorn1988faster}  and $\Omega(\log n)$ \cite{cheriyan2007packing}, respectively. Therefore, the tree-packing approach seems futile to us.

The same technical barrier was previously illustrated on the problem of computing a \emph{single} mincut. Karger's \cite{karger2000minimum} near-linear time global mincut algorithm was based on tree-packing, and it took 20 years before almost-linear algorithms for Steiner mincut and hypergraph mincut \cite{LP20,CQ21} were found using a very different technique, called \emph{isolating mincuts}. 

\paragraph{Second Technical Barrier: Minimal Isolating Mincuts Fails. }

The \emph{isolating mincuts }technique was recently discovered by \cite{LP20,abboud2021subcubic}. They show that given a graph $G=(V,E)$ and a terminal set $\T\subseteq V$, one can compute a $t$-mincut of $\T$ (i.e., a minimum cut separating $t$ from $\T\setminus t$) for all $t\in\T$ using logarithmic calls to max flow, instead of $|\T|$ many calls. In fact, they showed how to compute a \emph{minimal} $t$-mincut of $\T$ for all $t\in\T$, i.e., the unique mincut separating $t$ from $\T\setminus t$ that is ``closest'' to $t$. These cuts are called the \emph{minimal isolating mincuts of $\T$}. The technique instantly became very influential and found many applications \cite{abboud2021subcubic,li2021vertex,CQ21,li2021approximate,mukhopadhyay2021note,cen2022augmenting,li2022fair,li2022nearly,abboud2022apmf,AKLPST21}, many of which crucially exploit the minimal property of these isolating mincuts. 

Unfortunately, minimal isolating mincuts are ineffective for constructing a cactus: it cannot even distinguish a cycle from a clique! More precisely, consider a clique $K$ and a cycle $C$ with $n$ vertices. Scale the weight edges so that the weighted degree of each node of the two graphs agrees. All mincuts of $K$ consist of $n$ singletons cut, while all mincuts of $C$ contain all $\binom{n}{2}$ arcs of $C$. Now, for every vertex set $\T$, minimal isolating mincuts of $\T$ in both $K$ and $C$ are always a collection of singleton cuts. That is, minimal isolating mincuts alone fail to distinguish the mincut structures between the clique $K$ and the cycle $C$.

\paragraph{Our New Tool: Maximal Isolating Mincuts. }

The above counter-example naturally suggests we consider \emph{maximal }isolating mincuts. That is, given a terminal set $\T$, compute a maximal $t$-mincut of $\T$ for all $t\in\T$, which is the unique mincut separating $t$ from $\T\setminus t$ that is ``furthest'' from $t$. 

At first glance, maximal isolating mincuts seem unsuitable for almost-linear time algorithms because it is not even clear whether these cuts admit a near-linear space representation. In contrast, minimal isolating mincuts consist of disjoint vertex sets and so have linear size. 
For example, when $\T=\{s,t\}$, it is easy to see that the two maximal mincuts can overlap, and almost all vertices may be in both cuts. Therefore, it is conceivable that maximal isolating mincuts of $\T$ requires $\Omega(n|\T|)$ space, quadratic in the worse case.

Perhaps surprisingly, we show that these cuts' total size is linear and can also be computed using polylogarithmic calls to max flow. We devote \Cref{sec:max isocut} to proving this structural result. 

Let us reconsider the toy problem of ``cycle vs.~clique'' above. Indeed, maximal isolating mincuts can resolve this problem. Suppose we sample a terminal set $\T\subset V$ and then compute maximal isolating mincuts $\T$.
While all the minimal isolating mincuts in the clique remain singleton cuts, some maximal isolating mincuts will be non-singleton cuts in the cycle. So, in this simple example, we can distinguish the two graphs. 

\paragraph{Cactus via Maximal Isolating Mincuts.} It turns out that maximal isolating mincuts are powerful enough to capture \emph{all} mincuts.
By highly exploiting their structure, we show in \Cref{sec:steiner-cactus} how to construct Steiner cactus using polylogarithmic maxflow calls.  
At a high level, our algorithm significantly improves the divide-and-conquer approach by \cite{CX17} with linear recursion depth to only logarithmic.
To achieve this, the high-level idea is, in the divide step, our maximal-isolating-cut-based approach can \emph{split} the graph such that every part has size smaller by a constant factor, except at most one part that we guarantee that no more recursion is needed. (See 
\Cref{fig:case2} for an illustration.) Hence, the recursion depth is logarithmic.
In contrast, the algorithm of \cite{CX17} cannot certify this and might need to recurse on the biggest part for $\Omega(n)$ rounds.

\paragraph{First Challenge in Hypergraphs: Defining Maximal Isolating Mincuts}
Let us discuss the technical challenges in generalizing our techniques to hypergraphs.
First, it is tricky even to define the notion of maximal isolating mincuts for hypergraphs. A natural extension is a partition of vertices $(X, V\setminus X)$ that separates a specified terminal $t\in X\cap \mathcal{T}$ from other terminals such that $|X|$ is maximized while the number of boundary edges $|\partial X|$ is minimized.
However, the total output size of maximal isolating mincuts can be quadratic under this definition. Hence, it is useless for almost-linear time algorithms. 
For example, consider a hypergraph $G=(V, \{V\})$ with only one single hyperedge containing all the vertices. Designate half of the vertices as terminals, i.e., $\T\subset V$ and $|\T|=|V|/2$. In this case, every maximal isolating mincut has size $|V|-|\T|+1=\Omega(|V|)$.

It turns out the ``right'' definition is the following tweak --- we require that all mincuts $(X, V\setminus X)$ the algorithm is considering must be \emph{connected at the $X$ side}, that is, after removing the boundary edges $\partial X$, all vertices in $X$ must still be connected.
Under this definition, in \Cref{sec:max isocut hyper}, we show that all nice structural results we had on graphs transfer to hypergraphs. In particular, we can bound the total size of maximal isolating mincuts to be linear. As a sanity check, all maximal isolating mincuts in the above example are single vertices under the new definition, thereby the total output size becomes $O(|V|)$.

\paragraph{Second Challenge in Hypergraphs: Computing Hypercactus}
The second significant challenge is because a hypercactus contains hyperedges of rank higher than two.
Without very careful treatment, our divide-and-conquer algorithms, including previous algorithms by \cite{CX17},   
will split these higher-rank hyperedges of rank $r$ in the divide step into $\Omega(r)$ pieces.
Now, in the conquer step, the algorithm will need to merge them back and each merging step requires at least linear time (in fact, one max flow call) because we need to perform some test to know the topology of the gluing parts.
This incurs the running time  of $\Omega(pr)=\Omega(pn)$ where $p$ is the total size of the hypergraph, which is too slow.

We completely bypass this difficulty showing that our divide-and-conquer algorithm simply never splits these higher-rank hyperedges in a non-trivial way! This is done by again exploiting the ``right'' definition of the maximal isolating mincuts described above.
Compared with the algorithm by \cite{CX17}, our final algorithm in \Cref{sec:divide and conquer hypergraph} is arguably simpler as we do not perform the test related to higher-rank hyperedges and runs in almost-linear time.

%% file: appendix.tex
\section{Omitted Proofs from \Cref{sec:max isocut}}
\subsection{Proof of \Cref{lem:disjoint-posi-modularity}}
\label{app:proof of dosjoint posi}

If $X_A\cup X_B=V$, then it implies that $\C(X_A) = \C(V\setminus X_A)=\C(X_B\setminus X_A)\ge \C(X_B)=\C(X_A\setminus X_B) \ge \C(X_A)$ since the complement of $X_A$ is a $B$-cut of $\T$ and the complement of $X_B$ is an $A$-cut of $\T$.
Suppose that $X_A\cup X_B\subsetneq V$. By posi-modularity we have
\[
\C(X_A) + \C(X_B) \ge \C(X_A\setminus X_B) + \C(X_B\setminus X_A).
\]
Since $A$ and $B$ are disjoint, $X_A\setminus X_B$ is an $A$-cut and $X_B\setminus X_A$ is a $B$-cut. Hence, we have $\C(X_A\setminus X_B)\ge \C(X_A)$ and $\C(X_B\setminus X_A)\ge \C(X_B)$. Combining with the posi-modularity we have $\C(X_A)=\C(X_A\setminus X_B)$ and $\C(X_B)=\C(X_B\setminus X_A)$ as desired.

\hfill $\square$

\subsection{Proof of \Cref{lem:steiner-modularity}}
\label{app:proof of nesting submod}

By the submodularity of cut value function in \Cref{lem:submod},
\[
    \C(X_A) + \C(X_B) \ge \C(X_A\cap X_B) + \C(X_A\cup X_B) ~.
\]
Since both $X_A$ and $X_A\cap X_B$ are $A$-cuts and $X_A$ is $A$-mincut, we have $\C(X_A)\le \C(X_A\cap X_B)$. Similarly, both $X_B$ and $X_A\cup X_B$ are $B$-cuts and $X_B$ is $B$-mincut implies $\C(X_B)\le \C(X_A\cup X_B)$. Combining with the submodularity, we have $\C(X_A) =  \C(X_A\cap X_B)$ and $\C(X_B) = \C(X_A\cup X_B)$. Therefore, $X_A\cap X_B$ is a $A$-mincut of $\T$, and $X_A\cup X_B$ is a $B$-mincut of $\T$.

\hfill $\square$

\section{Omitted Proofs from \Cref{sec:steiner-cactus}}

\subsection{Proof of \Cref{lem:preprocessing-hypergraph}}
\label{sec:proof-of-preprocessing}

We show how to perform preprocessing in hypergraph, which implies
\Cref{lem:preprocessing} as a corollary for normal graph.

\begin{lemma}[Preprocessing]\label{lem:preprocessing-hypergraph}
Given a hypergraph $G$ and a terminal set $\T$, there exists an algorithm such that, with probability $1-n^{-11}$ the algorithm outputs a partition of $\T$ such that $\lambda(u, v) = \lambda_G(\T)$ if and only if $u$ and $v$ belongs to different parts. This algorithm runs in $O(\log^2 n)\cdot\MaxFlow(O(n+m), O(p))$ time.
\end{lemma}

\begin{proof}
The algorithm is exactly the same as invoking one step in Algorithm 4 (\textsc{CutThresholdStep} with post-processing) in~\cite{li2021approximate} with the slightest modification.
The algorithm works as follows.
The algorithm first computes $\lambda_G(\T)$, the value of $\T$-Steiner mincut on $G$ using Chekuri and Quanrud's algorithm~\cite{CQ21}.
Then, the algorithm
samples each terminal vertex with different sampling rates $2^{-i}$ for all $i=0, 1, 2, \ldots, 2^{-\lceil\log|\T|\rceil}$.
For each sampled vertex set $\T_i$ from the sampling rate $2^{-i}$, the algorithm computes the \emph{minimal} isolating mincut using the Isolating Cut Lemma~\cite{LP20}, and keeps only the mincut having the same value as $\lambda_G(\T)$.
Repeat the sampling procedure for $\Theta(\log n)$ times for ensuring that with high probability, every vertex $v$ obtains a minimal $\T$-Steiner mincut that contains $v$, as long as the size of this mincut is at most $|\T|/2$.
Let $\T_{\mathrm{large}}$ be the set of all terminal vertices that does not obtain such a mincut. Then these $\T$-Steiner mincuts together with $\T_{\mathrm{large}}$ \emph{is} a partition of $\T$ that we are looking for.
\end{proof}

%% file: isocut_hypergraphs.tex
\section{Proof of Maximal Isolating Mincuts on Hypergraphs}
\label{sec:max iso mincut proof}

In this section, we prove \Cref{thm:max-min-iso-cut-hypergraph}.
First, we verify that The Nesting \& Submodularity property also holds in hypergraphs, and leads to the uniqueness of maximal and minimal $A$-mincuts of $\T$ in a hypergraph.

\begin{lemma}[Nesting \& Submodularity on Hypergraph]
\label{lem:steiner-modularity-hypergraph}
Let $G$ be a hypergraph and let $\T$ be the set of terminals.
 Consider two nonempty subsets $A$ and $B$ of terminals such that $A\subseteq B\subsetneq \T$.
Let $X_A$ (resp. $X_B$) to be any $A$-mincuts (resp. $B$-mincuts) of $\T$.
Then, $\C(X_A\cap X_B) = \C(X_A)$.
Respectively, $\C(X_A\cup X_B) = \C(X_B)$, and $X_A\cup X_B$ is a $B$-mincut of $\T$. $\hfill\square$
\end{lemma}

The difference of \Cref{lem:steiner-modularity-hypergraph} with \Cref{lem:steiner-modularity} in a normal graph is that we cannot say $X_A\cap X_B$ is $A$-mincut of $\T$ but only their cut values are the same, since $(G/A)[X_A\cap X_B]$ may not be connected and hence violates \Cref{def:A-mincut hypergraph}.

The proof of Nesting \& Submodularity Lemma is 
almost identical as in proof of~\Cref{lem:steiner-modularity}
since the cut value function $\C$ also preserves submodularity in hypergraph,
except that we need to argue $X_A\cup X_B$ is a $B$-mincut of $\T$ which is straightforward to verify that $(G/B)[X_A\cup X_B]$ is connected.
Fortunately, \Cref{def:A-mincut hypergraph} does not affect the desired properties of maximal and minimal $A$-mincuts.

\begin{lemma}\label{lem:unique-max-iso-hyperedge-cut}
For any $A\subseteq\T$, there exists a unique maximal (resp. minimal) $A$-mincut of $\T$.
\end{lemma}

\begin{proof}
Suppose by contradiction $X_A$ and $X'_A$ are two different maximal $A$-mincuts of $\T$. By Nesting \& Submodularity \Cref{lem:steiner-modularity-hypergraph}, $X_A\cup X'_A$ is also a $A$-mincut of $\T$, contradicting the maximality of $X_A$ and $X'_A$.

Suppose by contradiction $X_A$ and $X'_A$ are two different minimal $A$-mincuts of $\T$. Again by Nesting \& Submodularity \Cref{lem:steiner-modularity-hypergraph}, we have $\C(X_A\cap X'_A) = \C(X_A)$. Suppose $r_A$ is the vertex contracted from $A$ in $(G/A)[X_A\cap X'_A]$. Let $Y$ be the connected component in $(G/A)[X_A\cap X'_A]$ connected to $r_A$, and $X = Y\cup A\setminus \{r_A\}$ which is the corresponding set of vertices in $G$. Then the boundary edges $\partial X\subseteq \partial(X_A\cap X'_A)$. Therefore, $\C(X)\le \C(X_A\cap X'_A) = \C(X_A)$ implies that $X$ is a also a $A$-mincuts of $\T$, contradicting the minimality of $X_A$ and $X'_A$.
\end{proof}

The reason that we want $(G/A)[X_A]$ to be connected in $A$-mincut of $T$ is because this definition leads to Pairwise Intersection Only Lemma on hypergraphs. Recall that the proof of \Cref{lem:three-crossing-isolating-mincuts} in a normal graph aims to show that there is no edge between $X = X_A\cap X_B\cap X_C$ and $X_A \setminus X$, contradicts to the connectivity of $(G/A)[X_A]$.
The proof for hypergraph is a bit different, since there may exist a hyperedge connecting $X$ and $X_A\setminus X$. Fortunately, we will show that these hyperedges can only be cut edges of $X_A$ and hence removed by \Cref{def:A-mincut hypergraph} when we analyze the connectivity of $(G/A)[X_A]$.

The Disjoint \& Posi-modularity property also holds in a hypergraph, and serves as the key to prove Pairwise Intersection Only Lemma on hypergraph. We need the definition of relaxed $A$-mincut before showing the Disjoint \& Posi-modularity for hypergraph.

\begin{definition}
\label{def:relaxed A-mincut hypergraph}
Let $G=(V, E)$ be a hypergraph and $\T\subseteq V$.
For any proper subset of terminals $A\subsetneq \T$, a cut $X$ is a \emph{relaxed $A$-cut} of $T$ if it satisfies the first conditions of $A$-cut in \Cref{def:A-mincut hypergraph}, i.e., the cut $(X, V\setminus X)$ separates $A$ and $\T\setminus A$, with $A\subseteq X$. A \emph{relaxed $A$-mincut} of $\T$ is a minimum valued relaxed $A$-cut of $\T$.
\end{definition}

\begin{lemma}
\label{lem:relaxed A-mincut equals A-mincut}
    Let $G=(V, E)$ be a hypergraph and $\T\subseteq V$.
    For any proper subset of terminals $A\subsetneq \T$, the value of $A$-mincut equals the value of relaxed $A$-mincut.
\end{lemma}
\begin{proof}
    The ``$\ge$'' direction is straightforward. For the other direction the proof is similar to the uniqueness of minimal $A$-mincut. We suppose a contradiction that there exists a relaxed $A$-mincut $X_A$ such that $\C(X_A)$ is strictly smaller than the value of $A$-mincut. Suppose $r_A$ is the vertex contracted from $A$ in $(G/A)[X_A\cap X'_A]$. Let $Y$ be the connected component in $(G/A)[X_A\cap X'_A]$ connected to $r_A$, and $X = Y\cup A\setminus \{r_A\}$ which is the corresponding set of vertices in $G$. Then the boundary edges $\partial X\subseteq \partial(X_A)$. Therefore, $\C(X)\le \C(X_A)$ which is smaller than the value of $A$-mincut, contradicting the definition of $A$-mincut.
\end{proof}

The proof of Disjoint \& Posi-modularity Lemma is the same as the normal graph since the cut value function $\C$ also preserves posi-modularity in a hypergraph.

\begin{lemma}[Disjoint \& Posi-modularity]\label{lem:disjoint-posi-modularity-hypergraph}
Let $A, B\subseteq \T$ be two nonempty subsets of terminals with $A\cap B=\emptyset$. Let $X_A$ (resp. $X_B$) be an relaxed $A$-mincut (resp. relaxed $B$-mincut) of $\T$. Then, $X_A\setminus X_B$ is a relaxed $A$-mincut of $\T$, and $X_B\setminus X_A$ is a relaxed $B$-mincut of $\T$. \hfill $\square$
\end{lemma}

\begin{lemma}[Pairwise Intersection Only on Hypergraph]
\label{lem:three-crossing-isolating-mincuts-on-hypergraph}
Given a hypergraph $G$, let $A, B, C\subseteq\T$ be three disjoint nonempty subsets of terminals. 
Let $X_A, X_B, X_C\subseteq V$ be any $A$-mincut of $\T$, $B$-mincut of $\T$, and $C$-mincut of $\T$ respectively.
Then $X_A\cap X_B\cap X_C=\emptyset$.
\end{lemma}

\begin{proof}
The proof is only interesting whenever the intersection of any two isolating mincuts is non-empty (i.e. crossing). Thus, without loss of generality, we assume that $X_A\cap X_B\neq\emptyset$, $X_B\cap X_C\neq\emptyset$, and $X_C\cap X_A \neq \emptyset$.

Define $X'_A = (X_A\setminus X_B)\setminus  X_C$, $X'_B = (X_B\setminus X_C)\setminus X_A$ and $X'_C = (X_C\setminus X_A)\setminus X_B$. These sets are non-empty since $A\subseteq X'_A$, $B\subseteq X'_B$, and $C\subseteq X'_C$. By posi-modularity (\Cref{lem:disjoint-posi-modularity-hypergraph}) we know that $X'_A$ (resp. $X'_B$ and $X'_C$) is relaxed $A$-mincuts (resp. $B$ and $C$).

Assume for contradiction that $X := X_A\cap X_B\cap X_C \neq \emptyset$.
We now aim to prove that there is no path from $X$ to $X'_A$ in the induced graph $(G/A)[X_A]$, which leads to a contradiction to \Cref{def:A-mincut hypergraph}. Similar as the proof for the normal graph, we consider the equality
\[
    \underbrace{(\C(X_A)+\C(X_B)+\C(X_C)) - (\C(X'_A)+\C(X'_B)+\C(X'_C))}_{\texttt{(LHS)}} = 0 ~.
\]
Note that $X'_A$ (resp. $X'_B$ and $X'_C$) is relaxed $A$-mincut (resp. $B$ and $C$) while $X_A$ (resp. $X_B$ and $X_C$) is $A$-mincut (resp. $B$ and $C$), the equality still holds since \Cref{lem:relaxed A-mincut equals A-mincut}.

Next, it suffices to prove the following combinatorial property.
\begin{enumerate}
    \item For any hyperedge $e$, the total contribution of the weight $e$ to the LHS of the equality is non-negative.
    \item For any hyperedge $e$ connecting $X$ such that $e$ contributes 0 to the LHS of the equality, either $e$ contributes to none of $\C(X_A), \C(X_B)$ and $\C(X_C)$, or $e$ contributes to all of $\C(X'_A), \C(X'_B)$ and $\C(X'_C)$.
\end{enumerate}
To see why it is sufficient, by first property, no hyperedge contributes positive to LHS of the equality. And by the second property, only these two kinds of hyperedges can be connecting to $X$. For the first case, $e$ contributes to none of $\C(X_A), \C(X_B)$ and $\C(X_C)$ means that $e$ only contains vertex in $X$. For the second case, $e$ is not in the induced graph $(G/A)[X_A]$ by definition of induced subhypergraph and hence $X$ is not connected with $X_A\setminus X$ in the induced graph, contradicts to the connectivity assumption in \Cref{def:A-mincut hypergraph}. It remains to prove the two properties.
\begin{enumerate}
    \item We further consider the cases that $e$ contributes to how many terms of $\C(X'_A), \C(X'_B)$ and $\C(X'_C)$. There is nothing to prove if $e$ contributes to none of them.
    
    If $e$ contributes to one of them, without loss of generality $\C(X'_A)$, then $e$ contains some vertex in $X'_A$, and also some vertex $v$ in $V\setminus X'_A = (V\setminus X_A)\cup X_B\cup X_C$. Therefore $v$ is in either $V\setminus X_A$, $X_B$ or $X_C$, hence also contributes to either $\C(X_A), \C(X_B)$ or $\C(X_C)$ respectively.
    
    If $e$ contributes to two of them, without loss of generality $\C(X'_A)$ and $\C(X'_B)$, then $e$ contains some vertex in $X'_A$ and in $X'_B$. So $e$ also contributes to $\C(X_A)$ and $\C(X_B)$.
    
    Otherwise $e$ contributes to all the three terms $\C(X'_A), \C(X'_B)$ and $\C(X'_C)$, then it is direct to see that $e$ also contributes to $\C(X_A), \C(X_B)$ and $\C(X_C)$.
    
    \item Similar to the proof of the first property, we consider the cases that $e$ contributes to how many terms of $\C(X'_A), \C(X'_B)$ and $\C(X'_C)$. This time $e$ must contain some vertex in $X$. 
    
    If $e$ contributes to none of them, the it also contributes to none of $\C(X_A), \C(X_B)$ and $\C(X_C)$ since the total contribution is 0, which satisfies the claim.
    
    If $e$ contributes to one of them, without loss of generality $\C(X'_A)$, then $e$ contains some vertex in $X'_A$. Since $e$ also contains vertex in $X$, it contributes to both $\C(X_B)$ and $\C(X_C)$, which gives a positive total contribution to LHS. So this case cannot happen.
    
    If $e$ contributes to two of them, without loss of generality $\C(X'_A)$ and $\C(X'_B)$, then $e$ contains some vertex in $X'_A$ and in $X'_B$. Since $e$ also contains vertex in $X$, it contributes to all of $\C(X_A), \C(X_B)$ and $\C(X_C)$, which also gives a positive total contribution to LHS. So this case cannot happen.
    
    Otherwise $e$ contributes to all the three terms $\C(X'_A), \C(X'_B)$ and $\C(X'_C)$, which satisfies the claim.\qedhere
\end{enumerate}
\end{proof}

\paragraph{Bounding the output size.} The total size of all maximal $v$-isolating mincuts on hypergraph is also $O(n)$, as a consequence of \Cref{lem:three-crossing-isolating-mincuts-on-hypergraph}. This results and its proof is the same as normal graph:

\begin{lemma}\label{lem:total-size-max-min-iso-cut-hypergraph}
Let $G$ be a hypergraph and $\T$ be a set of terminals.
For each $v\in \T$, let $X_v$ be the maximal $v$-isolating mincut.
Then, $\sum_{v\in \T} |X_v| \le 2n$.\hfill$\square$
\end{lemma}

\paragraph{A Divide and Conquer Algorithm.}

The maximal isolating mincut algorithm is exactly the same as \Cref{alg:max-min-iso-cut}, and we shall reduce computing the max-flow and $s$-minimal (resp. maximal) $s$-$t$ mincut on hypergraph to compute the max-flow on normal graph.

\paragraph{Compute $s$-Maximal $st$-Mincut on a Hypergraph via Max-Flow.}
Given a hypergraph $G = (V,E)$, a source $s$ and a sink $t$, the hypergraph can also be viewed as a bipartite graph $G' = (U',V',E')$ such that $U' = V, V' = E$ and $(u,e)\in E'$ iff $u\in e$. The flow network in the hypergraph $G$ can be viewed as flow network in the bipartite graph $G'$, with vertex capacity on $e\in V'$. So we can compute the max-flow of $G$ via computing the vertex capacity max-flow of graph $G'$.

The $s$-maximal $s$-$t$ mincut can be obtained by a post-processing of max-flow. The algorithm first computes the $s$-$t$ maxflow, and then examines the residue flow graph $G'^{(f)}_{\mathrm{flow}}$: the $s$-minimal $s$-$t$ mincut is the set of vertices in $U$ reachable from $s$, and the $s$-maximal $s$-$t$ mincut can be obtained by first computing $X\subset U$ to be the set of vertices not reachable to $t$, then find the connected component of the induced subhypergraph $G[X]$ connected by $s$.

The correctness proof of \Cref{alg:max-min-iso-cut} follows the same proof of \Cref{lem:max-iso-mincut-correctness}, by replacing the Nesting \& Submodularity Lemma using the hypergraph version \Cref{lem:steiner-modularity-hypergraph}.

\begin{lemma}[Correctness]\label{lem:max-iso-mincut-correctness-hypergraph}
Fix a terminal vertex $v\in \T$. There is a unique (leaf) subproblem where the maximal isolating mincut for $v$ is computed.
Let $\hat{X}_v$ be the cut returned at Line~\ref{line:run-max-flow} of \Cref{alg:max-min-iso-cut} for vertex $v$.
Then $\hat{X}_v=X_v$ is the maximal $v$-isolating mincut on the hypergraph $G$. $\hfill\square$
\end{lemma}
\thatchaphol{add proof.}

\begin{lemma}[Runtime]\label{lem:max-iso-mincut-runtime-hypergraph}
\textsc{MaxIsoMincut}$(G, \T)$ runs in time $O(\log |\T|)\cdot \mathrm{MaxFlow}(O(n+m), O(p))$ time on hypergraph $G$.
\end{lemma}

\begin{proof}
It suffices to bound the sum of graph sizes in all subproblems throughout the execution of \Cref{alg:max-min-iso-cut}.
First of all, the maximum depth of the recursion tree is $\left\lceil \log|\T| \right\rceil $ since in each recursive call the number of non-pivot terminals is reduced to half.
In addition, the number of subproblems in each recursion depth $i$ is at most $2^i\le |\T|$.

Now, we focus on a particular recursion depth $i>0$. Let $\{(G_j, \T_j, r_j)\}$ ($r_j$ denote the pivot vertex in hypergraph to avoid ambiguity) be all the subproblems whose recursion depth is $i$.
We observe that all terminals except pivot never goes to two subproblems so the subsets $\T'_j := \T_j\setminus \{r_j\}$ are disjoint.
Moreover, by Lines~\ref{line:contract-ga}-\ref{line:contract-gb} we know that for each $j$, removing the pivot $r_j$ from $G_j$ it is exactly the maximal $\T'_j$-mincut of $\T$ on $G$.

Using the Pairwise Intersection Only Lemma on Hypergraph (\Cref{lem:three-crossing-isolating-mincuts-on-hypergraph}),
we are able to conclude that every vertex in the input graph $G$ occurs in at most two subproblems at recursion depth $i$.
Therefore, the total number of vertices across all subproblems at depth $i$ is $\sum_{j}|V(G_j)| \le 2n + 2^i\le 2n+|\T| \le 3n$.

Next we bound the total volume of edges across all subproblems at depth $i$. 

\begin{align*}
    \sum_{j} \sum_{e\in E(G_j)} |e|
    & = \sum_{j} \sum_{e\in E(G_j)} \big|\{u\in e\mid u\in V(G_j)\}\big| \\
    & \le \sum_{j} \sum_{e\in E(G_j)} 2\big|\{u\in e\mid u\in V(G_j)\setminus\{r_j\} \}\big| \\
    & \le 2 \sum_{u\in V(G)} 2d(u) \\
    & = 4p ~.
\end{align*}
The last inequality is because every vertex in the input graph $G$ occurs in at most two subproblems at recursion depth $i$.

Therefore, in each subproblem $(G',\T')$, where the graph $G'$ has $n'$ vertices and $m'$ edges with total volume $p'$, the algorithm computes maximal $A$-mincut (and $B$-mincut) of $\T$ in $\mathrm{MaxFlow}(O(n'+m'), O(p'))$ time. By summing up the runtime per recursion depth, we obtain an upper bound to the desired total runtime $O(\log|\T|)\cdot\mathrm{MaxFlow}(O(n+m), O(p))$.
\end{proof}

\begin{proof}[Proof of \Cref{thm:max-min-iso-cut-hypergraph}.]
\Cref{thm:max-min-iso-cut-hypergraph} follows directly by \Cref{alg:max-min-iso-cut}, \Cref{lem:max-iso-mincut-correctness-hypergraph}, and \Cref{lem:max-iso-mincut-runtime-hypergraph}.
\end{proof}

%% file: omit_hypergraph.tex
\section{Omitted Proofs from \Cref{sec:divide and conquer hypergraph}}

\label{sec:omit hypergraph}

\subsection{Proof from \Cref{thm:steiner-hypercactus-main-hypergraph}}

The proof is almost the same as \Cref{thm:steiner-cactus-main}, except that we need to bound the total volume of hyperedges instead of only the number of normal edges.

\paragraph{Runtime.}
To analyze the total volume of edges and hyperedges across all subproblems within the same recursion depth, we bound normal edges and hyperedges separately.

For normal edges (rank-2 edges), the argument is the same as \Cref{thm:steiner-cactus-main}. after computing an induced decomposition from a good split collection, the total number of edges is increased by at most $\sum_{i=1}^\ell |V(G_i)|\le 2n$ (notice that we charge the number of the newly generated edges in the last decomposed graph $(G_{\ell+1}, \T_{\ell+1})$ to the edges across each split). Hence, we know that at any recursion depth there are at most $m+2n\log|\T|$ edges in total.
Here $n$ refers to the total number of vertices in the input graph, and $m$ refers to the total number of vertices in the input graph.

For hyperedges of rank larger than 2, we use a potential method similar to the vertices number analysis of \Cref{thm:steiner-cactus-main}. Let $E' = \{e\in E, |e|\ge 3\}$ to be the set of hyperedges and $p'_G := \sum_{e\in E'} |e|$ to be the total volume of hyperedges. Define $\Phi(G) := 3p'_G - 6|E'(G)|$. By definition $p'_G \ge 3|E(G)|\ge 0$, so $\Phi(G)\ge 0$. Consider the induced decomposition $\{(G_i, \T_i)\}$ on a good split collection of size $k$.
By definition of simple refinements, we know that $\Phi(G, \T) = \sum_{i=1}^{k+1} \Phi(G_i, \T_i)$.\thatchaphol{explain this a bit more like in normal graphs.}
Therefore, the sum of all potentials within the same recursion depth can be upper bounded by the root problem's potential. Since the total volume of any subproblem $p'_{G'}\le 3p'_{G'} - 6|E'(G')| = \Phi(G')$ by definition, we conclude that the total volume of hyperedges across all subproblems at any particular recursion depth (or any collection of subproblems that are not related to each other) is at most $\Phi(G) \le 3p$. 
Here $p$ refers to the total volume of hyperedges in the input graph.

Finally, we add up the runtime needed per recursion depth.
Fix any recursion depth, for each subproblem $(G_j, \T_j)$, by~\Cref{lem:compute good splits hypergraph}   the runtime spent in Line~\ref{line:compute-good-collection} is at most  $O((\log^3 n)\cdot\MaxFlow(O(|V(G_j)|+|E(G_j)|), O(p_{G_j}))$, the runtime spent for Line~\ref{line:compute-decomposition} is linear in the graph size $O(|V(G_j)|+p_{G_j})$,
and by \Cref{lem:merge-cactus hypergraph} merging hypercactus takes $O(\log|\T_j|\cdot \MaxFlow(O(|V(G_j)|+|E(G_j)|), O(p_{G_j}))$ time. 
Hence, combining all the subproblems together and using the bound of total volume of edges,
the runtime of \Cref{alg:divide-and-conquer-framework} is $O(\log^4 n)\cdot\MaxFlow(O(n+m), O(p + n\log |T|))$.\thatchaphol{Add preprocessing time}

\paragraph{Correctness.} By~\Cref{lem:compute good splits hypergraph}, with probability $1-n^{-11}$ the returned collection is good in Line~\ref{line:compute-good-collection}.
By the same analysis in the proof of~\Cref{thm:steiner-cactus-main}, 
we know that there are at most $4|\T|$ subproblems.
Hence the algorithm makes at most $4|\T|$
invocations to~\Cref{lem:compute good splits hypergraph}.
With a union bound, we know that with probability $1-4n^{-10}$ the collections of splits from all subproblems are good.
Now, by applying the union bound again to \Cref{lem:merge-cactus hypergraph} we know that the returned hypercactus is a $\T$-Steiner hypercactus of $G$ with probability at least $1-4n^{-10}$.
Therefore, with another union bound we know that with probability $1-8n^{-10}$ \Cref{alg:divide-and-conquer-framework} correctly output a $\T$-Steiner hypercactus.\hfill$\square$

%% file: app.tex
\section{Applications}
\label{sec:applications}

\subsection{Steiner Connectivity Edge Augmentation Problem}

Given an unweighted undirected graph $G$, a terminal set $\T$, and a target edge connectivity $\tau$. The goal of the \emph{Steiner connectivity edge augmentation problem} is to find a set of edges to $G$ 
whose addition makes the value of a $\T$-Steiner mincut to be at least $\tau$.

In this subsection, we briefly describe how to solve the Steiner connectivity edge augmentation problem, i.e., proving the first part of \Cref{cor:edge-augmentation}.

\medskip
\noindent\textbf{\Cref{cor:edge-augmentation}.}\textit{
There are randomized almost-linear time algorithms that can w.h.p. compute
\begin{itemize}[noitemsep,nolistsep]
\item the optimal solution of the Steiner connectivity augmentation problem.
\end{itemize}
}
\medskip

Cen et al.~\cite{cen2022augmenting} gives an algorithm that solves the \emph{edge connectivity augmentation problem}, which is a special case to the Steiner connectivity edge augmentation problem with $\T=V$.
Their algorithm utilizes Isolating Cut Lemma~\cite{LP20} and constructs a special hierarchy of vertex subsets called \emph{extreme set tree}.
With the extreme set tree,
the algorithm follows the framework of Bencz\'ur and Karger~\cite{benczur2000augmenting} that solves the \emph{degree-constrained edge connectivity problem} (DECA) given an extreme set tree.
The framework~\cite[Section 3]{cen2022augmenting} consists of 3 phases:

\begin{enumerate}[itemsep=0pt]
\item Using external augmentation, transform the degree constraints $\beta(v)$ to \emph{tight} degree constraints $b(v)$ for all $v\in V$.
\item Repeatedly add an augmentation chain to increase connectivity to at least $\tau-1$.
\item Add a matching defined on the $\T$-Steiner cactus if the connectivity does not reach $\tau$.
\end{enumerate}

Notice that the first two phases can be easily constructed in the Steiner case.
The third phase requires a computation of a $\T$-Steiner cactus of $G$.
By our new Steiner cactus algorithm (\Cref{thm:cactus main}), we are able to obtain an almost-linear time algorithm for DECA, and hence solving the Steiner connectivity augmentation problem.

\subsection{Hypergraph +1-Steiner-Connectivity Augmentation Problem}

Given a hypergraph $G=(V, E)$ and a terminal set $\T\subset V$,
the goal of
the \emph{hypergraph +1-Steiner-connectivity augmentation problem} is to compute a minimum sized set $E'$ of rank-2 edges such that $\lambda_{G\cup E'}(\T) \ge \lambda_{G}(\T)+1$. The \emph{optimal value} is defined to be the minimum possible $|E'|$.
In this subsection we briefly describe the proof to the second part of \Cref{cor:edge-augmentation}.

\medskip
\noindent\textbf{\Cref{cor:edge-augmentation}.}\textit{
There are randomized almost-linear time algorithms that can w.h.p. compute
\begin{itemize}[noitemsep,nolistsep]
\item the optimal value of the hypergraph +1-Steiner-connectivity augmentation problem. 
\end{itemize}
}
\medskip

Cheng~\cite{Cheng99} provides a formula for computing the optimal value of the hypergraph +1-connectivity augmentation problem, the value can be computed in linear time once we obtain a cactus (for all global mincuts) of $G$.
Below, we describe Cheng's result in the Steiner setting, thus solving the hypergraph +1-Steiner-connectivity augmentation problem:

\begin{theorem}[\cite{Cheng99}]
Given $G$ be a hypergraph and $\T$ be a terminal sets, let $H$ be the irredundant $\T$-Steiner hypercactus of $G$. Let $\alpha_G(\T)$ be the size of the largest hyperedge in $H$, and $\beta_G(\T)$ be the number of degree 1 nodes in $H$. Then the optimal value of $+1$-Steiner-connectivity augmentation is
\[
    \max\left\{ \alpha_G(\T) - 1, \left\lceil \frac{\beta_G(\T)}{2} \right\rceil \right\} ~.
\]
\end{theorem}

Therefore, with our new Steiner hypercactus algorithm (\Cref{thm:hypercactus main}), the optimal value of the hypergraph +1-Steiner-connecitivty augmentation problem can be solved in almost-linear time.

\subsection{Incremental Algorithm for  Hypergraph Mincuts}

Given a sequence of (unweighted) hyperedges inserting to an initially empty hypergraph $G$, 
the \emph{incremental hypergraph mincut problem} requires 
the data structure to correctly maintain the $\lambda(G)$ subject to this sequence of insertions.

Gupta and Karmakar~\cite{gupta2019incremental} gives an algorithm that solves the incremental hypergraph mincut problem in $O(\lambda n)$ amortized update time.
In this subsection, we show that our hypergraph cactus algorithm (\Cref{thm:hypercactus main}) can be used for substituting the cactus construction step and thus improving the amortized update time to the algorithm.

\medskip
\noindent\textbf{\Cref{cor:incremental-hypergraph-mincut}.}\textit{
There is an algorithm that, given an unweighted hypergraph $G=(V,E)$ undergoing hyperedge insertions, maintains a mincut in time $\Ohat(\lambda)$ amortized update time where $\lambda$ denotes the mincut value at the end of the updates. 
}
\medskip

The algorithm by Gupta and Karmakar~\cite{gupta2019incremental} proceeds in phases.
A new phase starts whenever $\lambda(G)$ increases. Within a phase, they spend $T_{\textrm{cactus}}+ \Otil(n + \sum_{e}|e|)$ time where the sum is over all edges inserted in this phase. 
By using previous hypercactus construction of Chekuri and Xu's algorithm~\cite{CX17}, they obtained an $\Otil(\lambda(G)n)$ amortized update time where $\lambda(G)$ denotes the mincut value after all the updates. 
Now, if we use our almost-linear time construction, the time per phase is $\Ohat(p)$ where $p$ is the total size of the hypergraph at the end of the phase. After $\lambda(G)$ phases, the total update time is then $\Ohat(p \lambda(G))$, which implies $\Ohat(\lambda(G))$ amortized update time.